\documentclass[a4paper,UKenglish,cleveref,autoref,english]{lipics-v2019}
\usepackage[utf8]{inputenc}
\usepackage[stmaryrd,textwidth=2.5cm,theorems=false,paralist=false,todo=show]{generic}
\usepackage{thmtools,thm-restate}
\usepackage{tikz}

\title{The monitoring problem for timed automata}
\titlerunning{}
\author{Alejandro Grez}{Pontificia Universidad Cat\'olica de Chile, Chile}{ajgrez@uc.cl}{}{}
\author{Filip Mazowiecki}{Max Planck Institute for Software Systems, Germany}{filipm@mpi-sws.org}{}{}
\author{Micha\l{} Pilipczuk}{University of Warsaw, Poland}{michal.pilipczuk@mimuw.edu.pl}{}%
       {This work is a part of project TOTAL that has received funding from the 
        European Research Council (ERC) under the European Union's Horizon 2020 
        research and innovation programme, grant agreement No.~677651.}
\author{Gabriele Puppis}{University of Udine, Italy}{gabriele.puppis@uniud.it}{}{}
\author{Cristian Riveros}{Pontificia Universidad Cat\'olica de Chile, Chile}{cristian.riveros@uc.cl}{}{}

\authorrunning{A. Grez, F. Mazowiecki, M. Pilipczuk, G. Puppis and C. Riveros}
\Copyright{Alejandro Grez, Filip Mazowiecki, Micha\l{} Pilipczuk, Gabriele Puppis and Cristian Riveros}

\ccsdesc[500]{Theory of computation~Models of computation}
\keywords{timed automata, monitoring problem, data stream, dynamic data structure}
\category{}
\relatedversion{}
\supplement{}
\nolinenumbers
\EventEditors{???}
\EventNoEds{2}
\EventLongTitle{???}
\EventShortTitle{???}
\EventAcronym{???}
\EventYear{???}
\EventDate{???}
\EventLocation{????}
\EventLogo{}
\SeriesVolume{??}
\ArticleNo{??}
\newcommand{\gabriele}[2][]{\todo[color=cyan!30, #1]{{\bf G.} #2}\ignorespaces}

\theoremstyle{plain}

\newenvironment{question}{\smallskip\em}{\smallskip}

\newcommand{\wLt}{\widehat{\Lt}}
\newcommand{\rt}{\mathtt{root}}
\newcommand{\ft}{\lambda}
\newcommand{\clk}{\mathtt{y}}
\newcommand{\lst}{\mathtt{l}}
\newcommand{\lstr}{\mathtt{r}} % was \mathtt{lr}
\newcommand{\lret}{\mathtt{lret}}
\newcommand{\frs}{\mathtt{f}}
\newcommand{\Kt}{K}
\newcommand{\Lt}{L}

\newcommand{\within}{\;\operatorname{WITHIN}\;}

\newcommand{\wAa}{\widehat{\Aa}}
\newcommand{\ww}{\widehat{w}}
\newcommand{\wq}{\widehat{q}}

\newcommand{\N}{\mathbb{N}}

\newcommand{\Oh}{\mathcal{O}}

\newcommand{\nat}{\mathbb{N}}
\newcommand{\real}{\mathbb{R}}
\newcommand{\realplus}{\mathbb{R}_{> 0}}
\newcommand{\set}[1]{\{#1\}}

\newcommand{\Aa}{\mathcal{A}}

\newcommand{\true}{\mathtt{true}}

\newcommand{\init}[1][]{\mathtt{init}_{#1}}
\renewcommand{\read}[1][]{\mathtt{read}_{#1}}
\newcommand{\accepted}[1][]{\mathtt{accepted}_{#1}}

\newcommand{\const}{\bar{c}}

\newcommand*{\eg}{e.g.\@\xspace}

\newcommand{\lnode}{\alpha}
\newcommand{\froot}{\beta}
\newcommand{\fnode}{\gamma}
\newcommand{\parent}{\mathtt{parent}}
\newcommand{\rank}{\mathtt{rank}}
\newcommand{\eltime}{\mathtt{timestamp}}
\newcommand{\aclock}{\mathtt{clock}}

\newcommand{\elchildren}{\mathtt{\#children}}
\newcommand{\elstates}{\mathtt{states}}
\newcommand{\elnext}{\mathtt{next}}
\newcommand{\elprev}{\mathtt{prev}}

\newcommand{\elrank}{\mathtt{rank}}

\newcommand{\node}{\mathtt{node}}

\usetikzlibrary{automata,chains,fit,shapes,calc}
\usetikzlibrary{decorations.pathreplacing}
\usetikzlibrary{patterns}
\usetikzlibrary{decorations.markings,decorations.pathmorphing,decorations.pathreplacing}
\usetikzlibrary{shadings}
\usetikzlibrary{shapes}
\usetikzlibrary{arrows}
\usetikzlibrary{calc}
\usepackage[outline]{contour}     % to have halo around text

\tikzset{double arrow/.style args={colored by #1 and #2}%
                                  {line width=3pt,#1, % first arrow
                                   postaction={draw,thin,#2}, % second arrow
                                  }}
\tikzstyle{dot} = [draw,shape=circle,fill, minimum size=1mm, inner sep=0pt, outer sep=0pt]
\tikzstyle{bigwhite} = [draw,shape=circle, minimum size=0.5cm, inner sep=0pt, outer sep=0pt]
\tikzstyle{biggray} = [draw,shape=circle, fill=gray!50, minimum size=0.5cm, inner sep=0pt, outer sep=0pt]
\tikzstyle{arrow} = [->, >=stealth, shorten >=1pt, shorten <=1pt, 
                     white, line width=2pt, postaction={draw,line width=0.5pt,black}]
\tikzstyle{harrow} = [->, >=stealth, shorten >=1pt, shorten <=1pt, 
                      white, line width=2pt, postaction={draw,line width=1pt,modernred}]
\tikzstyle{run} = [->, shorten >=1pt, shorten <=1pt, rounded corners=6, thick,
                   double arrow=colored by white and black]
\tikzstyle{cell} = [draw, minimum width=0.4cm, minimum height=0.4cm,
                    node distance=0.5cm,
                    inner sep=0, very thin, font=\vphantom{Ag}]
\tikzstyle{cellh} = [minimum width=0.4cm, minimum height=0.4cm,
                     node distance=0.5cm,
                     inner sep=0, very thin, font=\vphantom{Ag}]
\tikzstyle{circle} = [draw, shape=circle, minimum width=0.3cm, minimum height=0.3cm,
                    node distance=0.5cm,
                    inner sep=0, very thin, font=\vphantom{Ag}]
\tikzstyle{circlef} = [draw, shape=circle, fill=black, minimum width=0.3cm, minimum height=0.3cm,
                     node distance=0.5cm,
                     inner sep=0, very thin, font=\vphantom{Ag}]
\tikzstyle{circleh} = [shape=circle, minimum width=0.3cm, minimum height=0.3cm,
                     node distance=0.5cm,
                     inner sep=0, very thin, font=\vphantom{Ag}]
\tikzstyle{year} = [draw=gray, fill=gray!25, rounded corners=2mm]
\tikzstyle{task} = [draw=color2, fill=color2!25]

\begin{document}

\maketitle
% equation skip
\abovedisplayskip=6.5pt % equations skip
\belowdisplayskip=6.5pt % equations skip

% inlined equations with \< \>
\newcommand\inlineeqs{}  
\newcommand\ifinlineeqselse[2]{\ifdefined\inlineeqs #1\else #2\fi}

\makeatletter
\let\mydollar=$  
\makeatother
\def\<{\ifinlineeqselse{\mydollar}{\begin{equation*}}} 
\def\>{\ifinlineeqselse{\mydollar}{\end{equation*}}}

% itemized/enumerated lists that are inlined in short version
\newenvironment{itemize*}{\ifinlineelse{%
                            \let\olditem=\item%
                            \def\item{}%
                            \leavevmode\unskip%
                          }{%
                            \begin{itemize}
                          }}%
                         {\ifinlineelse{%
                            \let\item=\olditem%
                            \leavevmode\unskip%
                          }{%
                            \end{itemize}
                          }}
\newenvironment{enumerated*}{\ifinlineelse{%
                            \let\olditem=\item%
                            \def\item{}%
                            \leavevmode\unskip%
                          }{%
                            \begin{enumerated}
                          }}%
                         {\ifinlineelse{%
                            \let\item=\olditem%
                            \leavevmode\unskip%
                          }{%
                            \end{enumerated}
                          }}

% inlined paragraphs
\let\paragraph\undefined
\NewDocumentCommand\paragraph{s m}{\smallskip\par\noindent {\bfseries #2}~}

% there is already a subparagraph defined in LIPICS, but looks really hugly (and bolder than paragraph)
\renewcommand{\subparagraph}[1]{\smallskip\noindent\underline{\emph{#1}}~~\ignorespaces}

\begin{abstract}
We study a variant of the classical membership problem in automata theory, 
which consists of deciding whether a given input word is accepted by a given
automaton. We do so under a different perspective, that is, 
we consider a dynamic version of the problem, called {\em{monitoring problem}}, 
where the automaton is fixed and the input is revealed as in a stream, 
one symbol at a time following the natural order on positions. 
The goal here is to design a dynamic data structure that can be queried about 
whether the word consisting of symbols revealed so far is accepted by the automaton, 
and that can be efficiently updated when the next symbol is revealed.
We provide complexity bounds for this monitoring problem, 
by considering
timed automata that
process symbols interleaved with timestamps. The main contribution is that
monitoring of a one-clock timed automaton, with all its components but
the clock constants fixed, can be done in amortised constant time per 
input symbol.
\end{abstract}

\begin{picture}(0,0)
\put(392,25)
{\hbox{\includegraphics[width=40px]{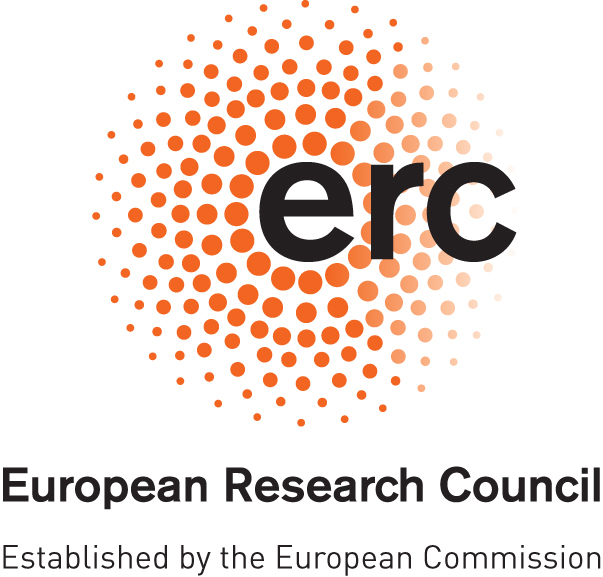}}}
\put(382,-35)
{\hbox{\includegraphics[width=60px]{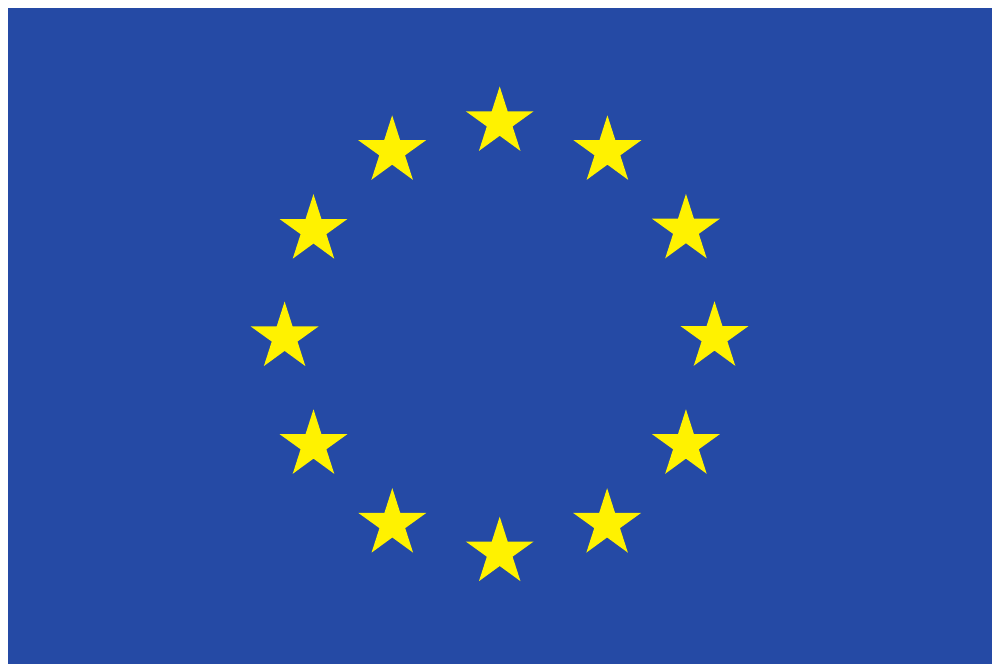}}}
\end{picture}

\newpage

\section{Introduction}
\label{sec:introduction}
% !TeX spellcheck = en_US

%Let $P \subseteq D^*$ be a property over a fixed domain $D$. Consider a sequence $s = d_1 d_2 \ldots$ of items over $D$ and assume that $s$ is read sequentially in a streaming fashion. The \emph{monitoring problem} asks to check, at any time~$i$, whether the prefix $d_1 d_2 \ldots d_i$ satisfies $P$ by taking the least amount of time per item (e.g. constant time). 
%
%The monitoring problem can be thought as a generalization of the classical membership problem in automata theory, which consists of deciding whether a finite input word is accepted by a given automaton. Namely, instead of verifying whether a finite sequence satisfies $P$, we want to check at any point in time whether the property $P$ holds when reading the (possible infinite) sequence from left to right. 
%For example, let $D = \Sigma$ be a finite set with two distinguished letters $a, b \in \Sigma$. One possible property $P_0$ can be given by the expression $D^* a D^* b D^*$. The the monitoring problem of $P_0$ reduced to decide, at any point in time, if we have read an $a$ followed by some $b$.

We consider the problem of \emph{monitoring} streams of data against a property. 
Precisely, the problem assumes that there is a fixed property $P\subseteq \Sigma^*$ 
over a data domain $\Sigma$ and an arbitrary stream $w = d_1 d_2 d_3 \dots \in \Sigma^\omega$,
and the goal is to tell, at each step $n$, whether the $n$-element prefix 
$d_1 d_2 \dots d_n$ of the stream verifies the property $P$.
For example, in the simpler form of the problem, the property $P$ 
may be a regular language given by a finite state automaton. 
In this case the monitoring problem could be thought of as a 
natural variant of the membership problem in automata theory.
In general, the monitoring problem requires to devise, for a given
property $P$, a \emph{dynamic data structure} that maintains the 
information whether the current stream prefix satisfies $P$ or not.
% and that can be updated consistently while consuming longer and longer 
% prefixes.
%In addition,
% \michal{Isn't this sentence essentially a repetition of the previous one?} 
% \gabriele{not sure, the stress now is on the ``efficiently'', while before was on the types of operations!}
This data structure should be updated whenever the next symbol of the input stream arrives, and we would like to perform these update operations
\emph{efficiently}: possibly in constant time, and thus independently of the 
input stream and of the amount of data consumed at any moment. 
Of course, the complexity of the operations 
may depend on the property $P$, but this is usually considered fixed, 
since the focus here is on the asymptotic complexity for monitoring streams 
of arbitrary length.

In this paper we consider the monitoring problem for ``regular''
properties enhanced with time constraints. Specifically, we consider 
properties given by \emph{timed automata}~\cite{AlurD94}, and streams of data that are
timed words, that is, sequences of letters from a finite alphabet interleaved
with timestamps. 
The main contribution is that monitoring of a one-clock timed automaton
$\Aa$ can be done in amortised constant time per consumed input element, 
assuming that $\Aa$ is fixed beforehand. 
In fact, we provide a finer complexity result, as we design
a dynamic data structure that supports monitoring with time
complexity that does not depend on the magnitude of the 
clock constants that appear 
in the transitions of $\Aa$.
% (this is formalized by saying that all the 
% components of $\Aa$ are fixed, except the clock constants). 
We also show that when the stream is discrete (when the timestamps
increase by exactly one unit of time) then monitoring
in constant (non-amortised) time is possible.

We complement these algorithmic results with a complexity lower bound, in which we consider
timed automata with multiple clocks and additive constraints~\cite{BerardD00}. 
We show that, subject to the {\sc{3SUM}} Conjecture,
the resulting family of properties cannot 
be monitored in amortised constant time, and not even in amortised strongly sub-linear time.
The {\sc{3SUM}} Conjecture is among the most popular hypotheses considered in computational geometry and in fine-grained complexity theory; see \eg an overview in~\cite[Appendix~A]{AbboudWY18}.

%As one is reading the sequence, one would expect that 
%the last data items become more relevant in our monitoring task than the initial ones. 

The considered timed variant of the monitoring problem turns out to be particularly
meaningful in the context of data streaming algorithms~\cite{babcock2002models, datar2002maintaining} 
and complex event recognition and processing~\cite{cugola2012processing}. There the input 
stream is produced by some sensors, or other similar systems, 
and the importance of each data item rapidly decays 
over time. In this setting, one would expect that the relevant data comes with an 
implicit time horizon, implying that the properties to be verified can be often
relativised to one or more time windows. 
% \gabriele{Cristian, please check if this makes sense (I somehow reused and simplified the old example)}
For example, in complex event processing, one can imagine a specification language
that extends classical regular expressions with constructs of the form
$E\ \mathtt{within}\ X$, requiring that a factor of the stream exists 
that satisfies the regular expression $E$ and that spans over at most $X$ time units.
Some properties are easily captured by one-clock timed automata, and thus can be efficiently 
monitored using our dynamic data structure.

%In order to model time for these contexts, consider $D = \Sigma\cup\realplus$ and $s$ to be a sequence of the form $t_1 d_1 t_2 d_2 \ldots$ with $d_i \in \Sigma$, $t_i \in \realplus$, and $t_i < t_{i+1}$ for every $i$. Here, each $t_i$ is the timestamp (e.g. in milliseconds) of $d_i$, representing the time when $d_i$ was produced by a sensor or system. 
%In these practical scenarios, we want to still monitor $P_0$ over $s$ but now only in the last $k$ milliseconds, for some input $k$. This can be naturally model by a property $P_1$ that contains all finite sequences $t_1 d_1 \ldots t_n d_n$ such that $d_i \ldots d_n$ satisfies $P_0$ and $t_n - t_i \leq k$ for some $i < n$. 
%Namely, monitoring $P_1$ over $s$ will report true after reading $d_n$ only if the data in the last $k$ milliseconds satisfies~$P_0$.

\paragraph{Related work.} 
% \michal{Changed the phrasing a bit here, please check.}
The monitoring problem we study here somewhat resembles the context of {\em{streaming algorithms}}; see e.g.~\cite{babcock2002models,datar2002maintaining,henzinger1998computing} for works with a similar motivation.
In this context, a typical problem is to compute (possibly approximately) some statistic or aggregate over a data sequence, where the main point is to assume a severe restriction on the space usage. 
Note that in our setting,
we focus on obtaining low time complexity of updates and queries, rather than studying the space complexity, so our work leans more towards the area of dynamic query evaluation~\cite{berkholz2017answering,idris2019efficient}. 
For Boolean properties (like in our monitoring problem), several papers~\cite{lewis1965memory,magniez2014recognizing,babu2013streaming} have considered streaming algorithms 
for testing membership in regular and context-free languages. 
Another variant of the problem was considered in~\cite{GanardiHL16,GanardiHKLM18,ganardi19}, 
where the regular property is verified on the last $N$ letters of the stream
instead of the entire prefix up to the current position. 
% Our approach is orthogonal to this research, as we consider the more general case 
% of timed automata and we focus on time complexity per update\filip{I don't understand the last remark. Where does it say that others don't ``focus on time complexity per update''? Unless this refers to the first part of the paragraph, then I would put this instead of the sentence ``Note that in our setting...'' because its confusing}.

% Finally, the work in~\cite{tangwongsan2017low} studies a monitoring problem for data streams 
% and provides an algorithm that takes constant time per input symbol. The algorithm computes 
% aggregates over a sliding window and it is not clear how to generalize it for timed automata 
% with one-clock.
% To the best of our knowledge, the present paper is the first work that considers the monitoring 
% problem for timed automata. 
The closest to our setting is the work~\cite{tangwongsan2017low}, 
which studies the monitoring problem for finite automata over a sliding window.
The authors provide a structure that takes constant time to update. 
We explain in Section~\ref{sec:problem} that this is a special case of our results for timed automata. 
In Appendix~\ref{appendix:examples} we give other application examples of our result.

\section{Preliminaries}
\label{sec:preliminaries}
\paragraph*{Finite automata.}
A \emph{finite automaton} is a tuple $\Aa = (\Sigma, Q, I, E, F)$, 
where $\Sigma$ is a finite alphabet, $Q$ is a finite set of states, 
$E \subseteq Q \times \Sigma \times Q$ is a transition relation, 
and $I,F \subseteq Q$ are the sets of initial and final states.
A run of $\Aa$ on a word $w = a_1 \ldots a_n \in \Sigma^*$ 
is a sequence
\< 
  \rho ~=~ q_0 \trans{a_1} q_1\trans{a_2} \ldots \trans{a_n} q_n, 
\>
where $(q_{i-1}, a_i, q_i) \in E$ for all $i=1,\dots,n$.
Moreover, $\rho$ is a \emph{successful} run if $q_0 \in I$ and $q_n \in F$.
The language \emph{recognized} by $\Aa$ is the set 
$\sL(\Aa) = \{w \in \Sigma^*\ \colon\ \text{there is a successful run of $\Aa$ on $w$}\}$.

\paragraph*{Timed automata.}
Let $X$ be a finite set of clocks, usually denoted by $\mathtt{x},\mathtt{y},\ldots$.
A \emph{clock valuation} is any function $\nu : X \to \real_{\ge0}$ from clocks to
non-negative real numbers.
\emph{Clock conditions} are formulas defined by the grammar
\<
C_X ~:=~ \true \mid 
         \mathtt{x} < c \mid \mathtt{x} > c \mid \mathtt{x} = c \mid 
         C_X \land C_X \mid C_X \lor C_X,
\>
where $\mathtt{x} \in X$ and $c \in \real_{\ge0}$. % \cup \{\infty\}$, and $\infty$ 
%is a dummy element that is assumed to be greater than all real numbers.
By a slight abuse of notation, we denote by $C_X$ the set of all clock conditions over the clock set $X$.
Given a clock condition $\gamma$ and a valuation $\nu$,
we say that $\nu$ \emph{satisfies} $\gamma$, and write $\nu \models \gamma$, 
if the arithmetic formula obtained from $\gamma$ by substituting each clock $\mathtt{x}$ 
with its value $\nu(\mathtt{x})$ evaluates to true. 
%\gabriele{added this}
%Note that clock conditions of the form $\mathtt{x} < \infty$ hold 
%vacuously and one could replace them by the propositional constant $\true$. 
%The choice of allowing the constant $\infty$ will help us simplifying some
%definitions in the following sections.

% \cristian{TODO: We used $I$ for the initial states and $I$ for intervals.}
% \gabriele{Done (replaced intervals $I$ by $J$)}
A \emph{timed automaton} is a tuple $\Aa = (\Sigma,Q,X,I,E,F)$, where
$Q$, $\Sigma$, $I$, $F$ are defined exactly as for finite automata, 
$X$ is a finite set of clocks, and 
$E \subseteq Q \times \Sigma \times C_X \times Q \times 2^X$ is a transition relation.
We say that $c\in\bbR_{\ge0}$ is a \emph{clock constant} of $\Aa$ if $c$ appears in some
clock condition of a transition from $E$.
A \emph{configuration} of $\Aa$ is a pair $(q, \nu)$, where $q \in Q$ 
and $\nu$ is a clock valuation. 
Recall that finite automata process words over a finite alphabet $\Sigma$; likewise, 
timed automata process timed words over an alphabet of the form $\Sigma\uplus\realplus$, 
with $\Sigma$ finite.%
\footnote{In the literature, timed words are often defined as words 
          over $\Sigma \times \realplus$, 
          meaning that every letter is tagged with a timestamp.
          Here, instead, we define a timed word as a sequence of letters
          possibly interleaved with time spans from $\realplus$.
          In this case, the timestamp of a letter is implicitly determined by 
          the sum of the time spans before it. 
          The current presentation enables a simpler definition of transitions.}
%          Moreover, the words in $\Sigma \times \realplus$ can be simulated in our setting by restricting to words of the form $(\Sigma \cdot \realplus)^*$.
%\gabriele{Or rather $\realplus\Sigma$? Anyway this was misleading}
Given a timed automaton $\Aa$ and a timed word $w = e_1 \ldots e_n \in (\Sigma \cup \realplus)^*$,
a run of $\Aa$ on $w$ is any sequence
\<
  \rho ~=~ (q_0,\nu_0) \trans{e_1} (q_1,\nu_1)\trans{e_2} \ldots \trans{e_n} (q_n,\nu_n),
\>
where each $(q_i, \nu_i)$ is a configuration and
\begin{itemize}
 \item if $e_i \in \realplus$, then $q_{i+1} = q_i$ and 
       $\nu_{i+1}(\mathtt{x}) = \nu_{i}(\mathtt{x}) + e_i$ for all $\mathtt{x} \in X$;
 \item if $e_i \in \Sigma$, then 
%       there are a clock condition $\gamma \in C_X$ and a 
%       set of clocks $Z \subseteq X$ such that
%       $(q_i,e_i,\gamma,q_{i+1},Z) \in E$, $\nu_i \models \gamma$, and
%       $\nu_{i+1}(\mathtt{x})$ is either $0$ or $\nu_{i}(\mathtt{x})$ 
%       depending on whether $\mathtt{x}\in Z$ or $\mathtt{x}\in X\setminus Z$.
       there is a transition $(q_i,e_i,\gamma,q_{i+1},Z) \in E$
       such that $\nu_i \models \gamma$ and 
       either $\nu_{i+1}(\mathtt{x})=0$ or $\nu_{i+1}(\mathtt{x})=\nu_{i}(\mathtt{x})$
       depending on whether $\mathtt{x}\in Z$ or $\mathtt{x}\in X\setminus Z$.
\end{itemize}
Thus, the set $Z$ in a transition $(q_i,e_i,\gamma,q_{i+1},Z) \in E$ corresponds to the subset of clocks that are reset when firing the transition.
Note that the values of the other clocks stay unchanged.

A run $\rho$ as above is \emph{successful} if $q_0\in I$, $\nu_0(\mathtt{x})=0$ 
for all $\mathtt{x}\in X$,
and $q_n \in F$.
The language \emph{recognised} by $\Aa$ is the set 
$\sL(\Aa) = \{w \in (\Sigma\cup\realplus)^*\ \colon\ \text{there is a successful run of $\Aa$ on $w$}\}$.

\paragraph{Size of an automaton.}
The size of a finite automaton $\Aa = (\Sigma, Q, I, E, F)$ is 
defined as $|\Aa| = |Q| + |E|$.
This is asymptotically equivalent to essentially every possible definition of size of a
finite automaton that can be found in the literature.
The size of a timed automaton $\Aa = (\Sigma,Q,X,I,E,F)$ is instead defined 
as $|\Aa| = |Q| + |X| + \sum_{(p,a,\gamma,q,Z) \in E} |\gamma|$, where
$|\gamma|$ is the number of atomic expressions (i.e.~expressions of the 
form $\mathtt{x} < c$, $\mathtt{x} > c$, $\mathtt{x} = c$) appearing
in the clock condition $\gamma$.
\emph{Note that the size of a timed automaton does not take into account the magnitude 
of the clock constants that appear in its clock conditions.}
The reader may consider these values as external parameters, provided to the 
monitoring algorithm on initialisation.

%Accordingly, in 
%this setting the time and space complexity of any algorithm working on timed automata
%should be independent of the magnitude (bit length) of the constants. 

\paragraph{Computation model.} 
In our setting, both the clock constants and the time spans read from the input stream can be arbitrary real numbers.
For this reason, we will use the real RAM model of computation, which is a model widely adopted in computational geometry and in data structures handling floating-point data.
In this model, we have integer memory cells that can store integers and floating-point memory cells that can store real numbers. There are no bounds on the bit length or precision of the stored numbers.
We assume that all basic arithmetic operations --- negation, addition, subtraction, multiplication, and division --- can be performed in unit time.
Note that modulo arithmetics and rounding is not included in the model; in fact, we will not use multiplication and division on real numbers either.

We remark that 
%it is straightforward to verify that 
if one assumes that all the numbers on input, 
such as clock constants or numbers appearing in the stream, are integers of bit length at most $M$,
% \gabriele{changed $w$ for $M$, since $w$ was usually the input word (I assume you did not want
% to relate $w$ and the bit length)}
then in our algorithms we may rely only on integers of bit length $\Oh(M+\log N)$, where 
$N$ is an upper bound on the total length of the stream that we expect to process.
Therefore, in this case we may assume the standard word RAM model with words of bit length $\Oh(M+\log N)$.

\section{The monitoring problem}
\label{sec:problem}
%\paragraph*{Problem.} 
In the algorithmic setting that we are going to consider, a timed automaton
processes an arbitrarily long and possibly infinite timed word, hereafter 
called \emph{stream}. 
It does so by consuming one element of the stream at a time (be it a letter 
or a time span), while updating its configuration. 
We would like to design a monitoring algorithm for an arbitrary stream 
processed by a fixed timed automaton $\Aa$, that readily answers the following query: 

\begin{question}
Does the time automaton accept the current prefix of the stream?
\end{question}

\noindent
In answering the following query, the monitoring algorithm can make use of a 
suitable data structure, that is initialised beforehand and is maintained 
consistent along the computation. Of course the goal here is to do so efficiently, 
that is, in amortised constant time assuming that $\Aa$ is fixed. We explain below 
what we mean precisely by this.

By saying that $\Aa$ is \emph{fixed} we mean that all its 
components (e.g.~input alphabet, sets of states, etc.) are fixed, with the only
exception of the constants that appear in the clock conditions of the transitions
of $\Aa$. Such constants are instead treated as formal parameters, 
and their actual values, represented by a tuple $\const$, are provided as input 
to the algorithm upon initialisation.
%\gabriele{added this, I think this is important by looking at the algorithm}
%In fact, we also assume that the relative order of the constants in $\Aa$ is 
%fixed; so when instantiating the formal parameters that act as constants for $\Aa$,
%one has to respect the fixed relative order. 
Note that the definition of the size $|\Aa|$ of $\Aa$ (cf.~Section \ref{sec:preliminaries}) 
reflects the above assumption, in the sense 
that it takes into account only the components of $\Aa$ that are considered 
to be fixed. 
% \gabriele{removed the comment about constants in the $O$-notation, since it was a duplicate}
%In particular, the ``constants'' hidden in the $\Oh(\cdot)$ 
%notation for the complexity, hereafter, will depend on $|\Aa|$, but not on 
%the clock constants that appear in the transitions of $\Aa$.

Once $\Aa$ is fixed, we can define the problem of {\em{monitoring}} $\Aa$ as 
the problem of designing a suitable dynamic data structure that supports certain 
operations 
%in low time complexity 
efficiently
upon reading the input stream.
Precisely, the data structure should support the following operations:
\begin{itemize}
 \item $\init(\const)$, that initialises the data structure for the timed automaton $\Aa$
        with a tuple $\const$ of clock constants;
%  \cristian{I added $\const$ as an input for the method $\init$, which makes explicitly that this is an input for our algorithm. The vector $\const$ is defined above when the constants are introduced.}
 \item $\accepted()$, that queries whether the prefix of the stream consumed up to the current moment
       is accepted by $\Aa$;
 \item $\read(e)$, that consumes the next element $e$ from the input stream, 
       be it a letter from $\Sigma$ or a time span from $\realplus$, 
       and updates the data structure accordingly.
\end{itemize}
%Note that $\init(\const)$ depends only on the constants $\const$ and the timed automaton $\Aa$, 
%since no part of the input stream has been consumed yet.
%
%The first operation that needs to be supported is denoted $\init$ and is meant 
%to initialize the data structure before any other operation could be executed.
%\gabriele{added this}\filip{and I would remove this sentence :P}
%In particular, this operation can be used, among other things, to bound the 
%clock constants of $\Aa$, seen as formal parameters, to some real values.
%Note that the $\init$ operation depends only on the timed automaton $\Aa$, 
%since no part of the input stream has been consumed yet.
%After the initialization, two operations, $\accepted$ and $\read (e)$, 
%can be executed repeatedly and in any order on the data structure. 
%The operation $\accepted$ queries the status of the monitor, telling whether 
%the prefix of the stream that has been consumed up to the current moment is 
%accepted by $\Aa$. 
%The operation $\read(e)$, instead, consumes the next element $e$ from the 
%input stream, be it a letter from $\Sigma$ or a time span from $\realplus$. 
%
The running time of each of these operations needs to be as low as possible,
possibly bounded by a constant that is independent of the input stream,
of the number of stream elements consumed so far, but also of the clock constants of $\Aa$.
However, as we assume that $\Aa$ is fixed, this constant running time may (and will) depend on $|\Aa|$.
Formally, we say that a data structure \emph{supports monitoring $\Aa$ in constant time} 
if the first operation $\init(\const)$ and every subsequent execution of 
$\accepted()$ or $\read(e)$ take time $\Oh(1)$.
Similarly, we say that a data structure
\emph{supports monitoring $\Aa$ in amortised time $f(n)$} 
if the first operation $\init(\const)$ and every subsequent execution of $\accepted()$ 
take time $\Oh(1)$, whereas if the execution of the $i$-th $\read(e)$ operation
takes time $t_i$, then for all $n$ we have
$\frac{1}{n} \cdot \sum_{i=1}^n t_i = \Oh(f(n))$.
In particular, we have \emph{amortised constant time} 
(resp.~\emph{amortised strongly sub-linear time}) 
if in the previous equation we can let $f(n)=1$ 
(resp.~$f(n)=n^{1-\delta}$ for any $\delta>0$).

We remark, for readers familiar with 
%the terminology of 
parametrised complexity,
that it may be convenient to consider the monitoring problem as a problem parametrised by the 
size of the fixed automaton $\Aa$.
%\footnote{Strictly speaking, in our setting, the size of a timed automaton does not take 
%          into account the magnitude of the constants that appear in its clock conditions.}
Then the ``constants'' hidden in the $\Oh(\cdot)$ notation are nothing 
%than 
but 
computable functions of the parameter $|\Aa|$.
This inscribes our work into the setting of {\em{dynamic parameterised algorithms}}, 
which is still a under-developed area; see e.g. the discussion in~\cite{AlmanMW17}.

\paragraph*{Results.}
In this paper we describe a data structure that supports monitoring
in either constant time or amortised constant time, depending on the form of
the input stream. More precisely, we say that a stream $w$ is \emph{discrete}
if its elements range over $\Sigma\uplus\{1\}$, that is, if all time spans coincide
with the time unit $1$. We will prove the following theorem:

\begin{restatable}{theorem}{OneClock}\label{th:one_clock}
Fix a timed automaton $\Aa$ with a single clock. 
% \gabriele{I have factored out the monitoring sentence, since it is more precise (the data
%           structure is always the same). Moreover I gained one line!}
There is a data structure that supports monitoring $\Aa$ 
\begin{itemize}
  \item in constant time on discrete streams,
  \item in amortised constant time on arbitrary streams.
\end{itemize}
\end{restatable}

\smallskip\noindent
The proof of Theorem~\ref{th:one_clock} is deferred to Section~\ref{sec:structure-new}.
We do not know whether this theorem can be generalised to 
timed automata with more than one clock without sacrificing the amortised 
constant time complexity of updates. 

On the other hand, we are able to establish a negative result for a slightly 
more powerful model of timed automata, called timed automata with 
additive constraints (see e.g.~\cite{BerardD00}).
Formally, a \emph{timed automaton with additive constraints} is defined 
exactly as a timed automaton -- that is, as a tuple $\Aa = (\Sigma,Q,X,I,E,F)$ 
consisting of an input alphabet, a set of states, a set of clocks, etc. 
--- but clock conditions are now allowed to satisfy an extended grammar obtained 
by adding new rules of the form $\big(\sum_{\mathtt{x}\in Z} \mathtt{x}\big) \sim c$, where 
$Z\subseteq X$ and $\sim \; \in \set{<, >, =}$.

Our lower bound further relies on a complexity theoretical assumption
related to the {\sc{3SUM}} problem.
Recall that in the {\sc{3SUM}} problem we are given a set $S$ of positive real numbers 
and the question is to determine whether there exist $a,b,c\in S$ satisfying $a+b=c$.
%\footnote{An alternative but equivalent formulation of {\sc{3SUM}} asks 
%          whether there are $a,b,c\in S$ such that $a+b+c=0$.}
It is easy to solve the problem in time $\Oh(n^2)$, where $n=|S|$. 
Our lower bound is based on the hypothesis that 
the quadratic running time for 3SUM cannot be significantly improved.

\begin{restatable}[{3SUM} Conjecture]{conjecture}{SumConjecture}\label{cnj:3sum}
In the real RAM model, the {\sc{3SUM}} problem cannot be solved 
in strongly sub-quadratic time, that is, in time $\Oh(n^{2-\delta})$ 
for any $\delta>0$, where $n$ is the number of numbers on the input.
\end{restatable}

\smallskip\noindent
The {\sc{3SUM}} Conjecture is among the most popular hypotheses considered 
in computational geometry and fine-grained complexity theory; see e.g. 
an overview in~\cite[Appendix~A]{AbboudWY18}.
It was introduced by Gajentaan and Overmars~\cite{GajentaanO95,GajentaanO12} 
in a stronger form, which postulated the non-existence of \emph{sub-quadratic}
algorithms, that is, achieving running time $o(n^2)$.
This formulation was refuted by Gr{\o}nlund and Pettie~\cite{GronlundP18}, who 
gave an algorithm for {\sc{3SUM}} with running time $\Oh(n^{2}/(\log n / \log \log n)^{2/3})$ 
in the real RAM model, which can be improved to $\Oh(n^{2}(\log \log n)^2/\log n)$ 
when randomisation is allowed. However, the existence of a strongly sub-quadratic 
algorithm is widely open and conjectured to be hard.

%We use the {\sc{3SUM}} Conjecture to derive a complexity lower bound for 
%a variant of our monitoring problem:
We can now state a lower bound for a variant of the monitoring problem:

\begin{restatable}{theorem}{TwoClocks}\label{th:two_clocks}
If the {\sc{3SUM}} Conjecture holds, then there
is a two-clock timed automaton $\Aa$ with additive constraints such that
there is no data structure supporting monitoring $\Aa$ 
in amortised strongly sub-linear time, and hence also not in amortised
constant time.
\end{restatable}

\smallskip\noindent
The proof of the above theorem is in the appendix, together with 
further discussion about the {\sc{3SUM}} Conjecture.
Again, we do not know whether a negative result similar the above 
one also holds for plain timed automata (without additive constraints).

\smallskip
Before presenting the data structure for Theorem~\ref{th:one_clock},
we show an application of this result. 
%The proof of Theorem~\ref{th:one_clock} is given in Section~\ref{sec:structure-new},
%while that of Theorem~\ref{th:two_clocks} is given in Section~\ref{sec:lower}. 

\paragraph{An application example.}
Here we consider only 
streams that are discrete. In fact, we will enforce a slightly more
restricted form for such streams: we will assume that they belong to
the language $(\{1\} \cdot \Sigma)^\omega$, namely, that the letters 
from $\Sigma$ are interleaved by the time unit~$1$.
%
%\begin{example}\label{ex:windows}
Let $\Aa = (\Sigma, Q, I, E, F)$ be a finite automaton and $C>0$ a number
defining the width of a window. We consider the membership problem in the 
\emph{sliding window model} (see, for instance \cite{ganardi19}), that is, we process 
an arbitrary input $w=a_1 a_2 a_3 \ldots$ over $\Sigma$, consuming one letter at 
a time from left to right, while maintaining the answer to the following query: 
is the sequence of the last $C$ consumed letters accepted by $\Aa$ or not?
Below, we explain how this problem can be reduced to our monitoring problem 
for timed automata and discrete streams. 

We map $w = a_1a_2a_3\ldots$ to a corresponding discrete
stream $\ww = 1 a_1 1 a_2 1 a_3 \ldots$, and modify $\Aa$ to obtain a 
corresponding timed automaton $\wAa$, as follows. We introduce a new state $\wq$, 
which will be the only final state in $\wAa$, and a clock $\mathtt{x}$. 
We then replace every transition $(q,a,q')$ of $\Aa$ with the transition 
$(q,a,\true,q',\emptyset)$. Note that these transitions have
a vacuous clock condition, hence they are applicable in $\wAa$ whenever 
the original transitions of $\Aa$ are so. 
In addition, when the former transition $(q,a,q')$ 
reaches a final state $q'\in F$, 
we also have a transition 
$(q,a,\mathtt{x}=C,\wq,\emptyset)$ in $\wAa$. 
%Intuitively, this construction strengthens the acceptance condition of $\Aa$
%by enforcing that the accepted word has length exactly $C$.
Finally, we add looping transitions on the initial states that reset the clock, 
that is, transitions of the form $(q,a,\true,q,\{\mathtt{x}\})$, with  
$q\in I$ and $a \in \Sigma$.
Figure~\ref{fig:window} shows the timed automaton $\wAa$ corresponding
to an automaton $\Aa$ recognising~$ab^*a$.

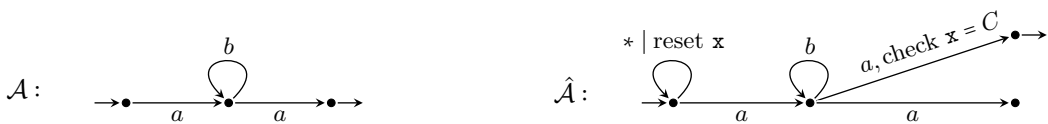
\begin{figure}
\vspace{-3mm}
\centering
\begin{tikzpicture}[scale=0.9]
\clip (-1.9,-0.3) rectangle (13.5,1.4);
\begin{scope}
  \draw (-1.5,0.2) node {$\Aa:$};
  
  \draw (0,0) node [dot] (p) {};
  \draw (1.5,0) node [dot] (q) {};
  \draw (3,0) node [dot] (r) {};

  \draw [arrow] (-0.5,0) to (p);
  \draw [arrow] (r) to (3.5,0);
  \draw [arrow] (p) to node [black, below=-0.05] {\small $a$} (q);
  \draw [arrow] (q) to node [black, below=-0.05] {\small $a$} (r);
  \draw [arrow] (q) to [out=45, in=135, looseness=35] node [black, above=-0.05] {\small $b$} (q);
\end{scope}
\begin{scope}[xshift=8cm]
  \draw (-1.5,0.2) node {$\hat\Aa:$};
  
  \draw (0,0) node [dot] (p) {};
  \draw (2,0) node [dot] (q) {};
  \draw (5,0) node [dot] (r) {};
  \draw (5,1) node [dot] (r') {};
  
  \draw [arrow] (-0.5,0) to (p);
  \draw [arrow] (r') to (5.5,1);
  \draw [arrow] (p) to [out=45, in=135, looseness=35] node [black, above=-0.05] {\small $*\mid\text{reset } \mathtt{x}$} (p);
  \draw [arrow] (p) to node [black, below=-0.05] {\small $a$} (q);
  \draw [arrow] (q) to node [black, sloped, above=-0.05] {\small $\qquad a, \text{check } \mathtt{x}=C$} (r');
  \draw [arrow] (q) to node [black, below=-0.05] {\small $a$} (r);
  \draw [arrow] (q) to [out=45, in=135, looseness=35] node [black, above=-0.05] {\small $b$} (q);
\end{scope}
\end{tikzpicture}

\caption{Reducing the sliding window membership problem to the monitoring problem.}\label{fig:window}
\vspace{-3mm}
\end{figure}

From the above construction it is clear that $\wAa$ accepts a prefix 
$1 a_1 \dots 1 a_n$ of $\ww$ if and only if $\Aa$ accepts the
$C$-letter factor $a_{n-C+1} \dots a_n$ of $w$. Thus, the acceptance
problem for $\Aa$ in the width-$C$ sliding window model is reduced to 
the monitoring problem for $\wAa$ over the stream $\ww$. 
By Theorem \ref{th:one_clock}, we know that there is a data structure
that supports monitoring in constant time. This means that we can process
one letter at a time from a word $w$, while answering in constant time 
whether $\Aa$ accepts the sequence of the last $C$ consumed letters.
Note that the complexity here is independent of the choice of $C$.
%\end{example}

We also remark that there is a similar reduction from the sliding window
problem to the monitoring problem, where the timed automaton $\wAa$ uses 
a clock condition, say $\mathtt{x}=1$, that is fixed prior to the choice 
of the window length, and where the input timed word is of the form 
$\delta a_1 \delta a_2 \delta a_3 \dots$, with $\delta = \frac{1}{C}$.
This intuitively shows that, even if we fix the clock constants 
together with the timed automaton, the monitoring problem does not trivialise.

\section{Amortised constant-time data structure} %: new write-up}
\label{sec:structure-new}
%\subsection{Overview and the outer data structure}

We prove Theorem~\ref{th:one_clock} by describing 
a data structure for monitoring a given timed automaton.

\paragraph*{Notation and definitions.}
Let us fix, once and for all, a timed automaton 
$\Aa = (\Sigma,Q,X,I,E,F)$ with a single clock $\mathtt{x}$.
By adding a non-accepting sink state, if necessary, we may assume that for every $q\in Q$ and $a\in \Sigma$, 
some transition over letter $a$ can be always applied at $q$.
Since $\Aa$ uses only one clock, every configuration of $\Aa$ can be written simply as a pair $(q,t)$, 
where $q\in Q$ is the state and $t\in \real_{\geq 0}$ is the value of the clock $\mathtt{x}$.

Note that by fixing $\Aa$ we also fix the number of constants
that appear in the clock conditions in $E$. Let us enumerate them as
\<
  0=C_0 < C_1 <\ldots <C_k,
\>  
where, without loss of generality, we assume that the clock constant $C_0=0$ is always present.
For simplicity we also let $C_{k+1}=\infty$, assuming that $t < \infty$ for every $t \in \real$.
As $\Aa$ is fixed, we consider $|\Sigma|$, $|Q|$, and $k$ as fixed constants.

%The transition function of $\Aa$ gives rise to a mapping $\Phi$ on sets of configurations defined as follows:
%for a set of configurations $X$ and a symbol $e\in \Sigma\cup \realplus$, we set
%$$\Phi(X,e) = \{\,(q,t)\ \colon\ \textrm{there exists} (p,s)\in X\textrm{ such that }(p,s)\trans{e}(q,t)\,\}.$$
%We lift this definition to words over $\Sigma\cup \realplus$ in the natural manner.

%It will be often convenient for us to operate on collections of configurations with the samce clock value.
%Let then a {\em{bundle}} be a pair $(X,t)$, where $X\subseteq Q$ is a non-empty set of states and $t\in \real_{\geq 0}$ is a clock value.
%Such a bundle {\em{represents}} the set of configurations $\{(q,t)\colon q\in X\}$.
%Note that every set of configurations can be uniquely decomposed into a set of bundle with pairwise different clock values.

Consider now an arbitrary stream $w \in (\Sigma \cup \realplus)^\omega$.
For every $n\in\bbN$, let $w_n = w[1\ldots n]$ be the $n$-element prefix of $w$.
Recall that $w_n$ can be thought of as the stream prefix that is disclosed 
after $n$ operations $\read(e)$.
%We define the \emph{time elapsed} in the stream prefix $w_n$ as
%$T_n = \sum_{i=1}^s w(j_i)$, where $w_{j_1}\ldots w_{j_s}$ 
%is the maximal subsequence of $w_n$ consisting of time spans in $\realplus$.
We say that a configuration $(q,t)$ is \emph{active} at step $n$ 
if there is a run of $\Aa$ on $w_n$ that starts in a configuration $(q_0,0)$ for some $q_0\in I$ and ends in $(q,t)$. 
We let $\Kt_n$ be the set of all configurations~$(q,t)$ that are active at step $n$.
%Note that $\Kt_0=\{(q,0)\colon q\in I\}$ and $\Kt_n=\Phi(\Kt_0,w_n)$ for all $n\geq 1$.

%We also say that a clock value $t\in\real_{\geq 0}$ is \emph{active} at step $n$ 
%if for some $q\in Q$, the configuration $(q,t)$ is active at step $n$.

\paragraph*{Partitioning the problem.}
It is clear that the monitoring problem essentially boils down to designing an efficient data structure that maintains $\Kt_n$ under reading consecutive elements from the stream.
This data structure should offer a query on whether $\Kt_n$ contains an accepting configuration.
The main observation is that configurations with clock values behaving in the same way with respect to the clock constants $C_1,\ldots,C_k$ satisfy exactly the same clock conditions in $E$.
Precisely, let us consider the partition of $\real_{\geq 0}$ into intervals
\begin{gather*}
J_0 = [C_0,C_0],\quad J_1 = (C_0,C_1),\quad J_2 = [C_1,C_1],\quad J_3 = (C_1,C_2),\quad J_4 = [C_2,C_2],\\
J_5 = (C_2,C_3),\quad \ldots \quad J_{2k-1} = (C_{k-1},C_k),\quad J_{2k} = [C_k,C_k],\quad J_{2k+1} = (C_k,C_{k+1}).
\end{gather*}
The following assertion is clear: for any two configurations $(q,t)$, $(q,t')$,
with $t,t'\in J_i$ for some $0\le i\le 2k+1$,
%such that $t$ and $t'$ belong to the same interval among $J_0,J_1,\ldots,J_{2k+1}$, 
exactly the same transitions
are available in $(q,t)$ as in $(q,t')$.

%\begin{align*}
%J_0 = & [C_0,C_1)& \qquad J_1 = & [C_1,C_1]\\
%J_2 = & (C_1,C_2)& \qquad J_3 = & [C_2,C_2]\\
%& \vdots & \vdots & \\
%J_{2k-2} = & (C_{k-1},C_k)& \qquad J_{2k-1} = & [C_k,C_k]\\
%J_{2k} = & (C_k,C_{k+1})& &
%\end{align*}

For $n\in \bbN$ and $i\in \{0,\ldots,2k+1\}$, let 
\<
  \Kt_n[i] = \{\,(q,t)\in \Kt_n\ \colon\ t\in J_i\,\}.
\>
%Thus $\{\Kt_n[i]\ \colon\ i\in \{0,\ldots,2k+1\}\}$ is a partition of $\Kt_n$. 
The idea is to maintain each set $\Kt_n[i]$ in a separate data structure.
Each of these data structures follows the same design, which we call the {\em{inner data structure}} 
and which is outlined below.

\newcommand{\Dt}{\mathbb{D}}
\newcommand{\inz}{\mathtt{init}}
\newcommand{\acc}{\mathtt{accepted}}
\newcommand{\ins}{\mathtt{insert}}
\newcommand{\updTime}{\mathtt{updateTime}}
\newcommand{\updLett}{\mathtt{updateLetter}}

\paragraph*{Inner data structure: an overview.}
Every inner data structure is constructed from an interval 
$J\in \{J_0,\ldots,J_{2k+1}\}$. 
We will denote it by $\Dt[J]$, or simply by $\Dt[i]$ when $J=J_i$.
Each structure $\Dt[J]$ stores a set of configurations $\Lt$ satisfying the following invariant: all clock values of configurations in $\Lt$ belong to $J$. 
% \gabriele{added this initialization, and removed the $\inz()$ method, to hide the fact that the inner data structure
%           should be initialized with a time interval (I think that 
%           having a notation that reflects this explicitly would only
%           complicate things).}
%Initially, the structure $\Dt[J]$ stores the empty set of configurations, i.e.~$\Lt=\emptyset$.
In the final design we will maintain the invariant that the set $\Lt$ stored by $\Dt[i]$ 
at step $n$ is equal to $\Kt_n[i]$, 
but for the design of $\Dt[J]$ it is easier to treat $\Lt$ as an arbitrary set of configurations
with clock values in $J$.

The inner data structure should support the following methods: 
% \gabriele{removed the $\inz()$ method}
\begin{itemize}
 \item Method $\inz(J)$ stores the interval $J$ and initialises 
       $\Dt[J]$ by setting $\Lt=\emptyset$.
 \item Method $\acc()$ returns true or false, depending on whether $\Lt$ contains an accepting configuration, that is, a configuration $(q,t)$ such that $q\in F$.
 \item Method $\ins(q,t)$ adds a configuration $(q,t)$ to $\Lt$. 
       This method will be always applied with a promise that 
       $t\in J$ and 
       $t\leq t'$ for all configurations $(q',t')$ already present in $\Lt$. 
 \item Method $\updTime(r)$, where $r\in \real_{>0}$, increments the clock values of all configurations in $\Lt$ by $r$.
       All configurations whose clock values ceased to belong to $J$ are removed from $\Lt$, and they are returned by the method on output. 
       This output is organised as a doubly linked list of configurations, sorted by non-decreasing clock values.
 \item Method $\updLett(a)$ updates $\Lt$ by applying to all configurations in $\Lt$ 
       all possible transitions over the given letter $a\in\Sigma$.
       Precisely, the updated set $\Lt$ comprises all configurations $(q,t)$ that can be obtained from configurations belonging to $\Lt$ before the update using transitions over $a$ that do not reset the clock.
       All configurations $(q,0)$ which can be obtained from $\Lt$ using transitions over $a$ that do reset the clock are not included in the updated set $\Lt$, but are instead returned by the method as a doubly linked list.
\end{itemize}
%Note that thus, in method $\updTime(r)$ the set $\Lt$ is updated to $\Phi(\Lt,r)$, except that all configurations whose clock values do not fit into $J$ do not get included, and are returned instead.
%Similarly, in method $\updLett(a)$ the set $\Lt$ is updated to $\Phi(\Lt,a)$, except that all configurations where the clock has been reset do not get included, and are returned instead.
%
In Section~\ref{sec:inner} we will provide an efficient implementation of the inner data structure, which is encapsulated in the following lemma.

\begin{lemma}\label{lem:inner}
For every $J=J_i$, $i\in \{0,1,\ldots,2k+1\}$, the inner data structure $\Dt[J]$ can be implemented so that methods $\inz()$, $\acc()$, $\ins(\cdot,\cdot)$,  and $\updLett(\cdot)$
run in constant time, while method $\updTime(\cdot)$ runs in time linear in the size of its output.
\end{lemma}

\noindent
%Note that there is no amortization involved in Lemma~\ref{lem:inner}. 
We postpone the proof of Lemma~\ref{lem:inner} and we show now how to use it to prove Theorem~\ref{th:one_clock}.
That is, we design an {\em{outer data structure}} that solves the monitoring problem for $\Aa$. % which uses several copies of the inner data structure as black-boxes.

\subsection{Outer data structure}\label{sec:inner}
%\paragraph*{Outer data structure.}

The outer data structure consists of a list of data structures $\Dt[0],\ldots,\Dt[2k+1]$, where each $\Dt[i]$ is a copy of the inner data structure constructed for the interval $J_i$.
We will maintain the following invariant:
\begin{enumerate}[{I}1.]
 \item\label{i:partition} After step $n$, for each $i\in \{0,1,\ldots,2k+1\}$ the data structure $\Dt[i]$ stores $\Kt_n[i]$.
\end{enumerate}

We first explain how the outer data structure implements the promised operations: initialisation, queries about the acceptance, and updates upon reading the next symbol from the stream $w$.
Then we discuss the amortised complexity of the updates.

\subparagraph{Initialisation.} 
Given the tuple $\const=(C_0,C_1,\dots,C_k)$ of clock constants of $\Aa$, 
with $C_0=0$, we initialise $2k+1$ copies $\Dt[0],\ldots,\Dt[2k+1]$
of the inner data structure
by calling method $\inz(J)$ for each interval $J$ among
% one for each of the intervals 
$J_0,J_1,\dots,J_{2k+1}$.
%We simply apply $\inz()$ to each of the data structures $\Dt[0],\ldots,\Dt[2k+1]$. 
Then, for each initial state $q$, we apply method $\ins(q,0)$ on $\Dt[0]$. 
As $\Kt_0=\{(q,0)\ \colon\ q\in I\}$, after this we have that 
Invariant~(I\ref{i:partition}) holds for $n=0$.

\subparagraph{Query.} We query all the data structures $\Dt[0],\ldots,\Dt[2k+1]$ for the existence of accepting configurations using the $\acc()$ method, and return the disjuction of the answers.
The correctness follows directly from Invariant~(I\ref{i:partition}).

\subparagraph{Update by a time span.} Suppose the next symbol read from the stream 
is a time span $r\in \realplus$. We update the outer data structure as follows.
First, we apply method $\updTime(r)$ to each data structure $\Dt[i]$. 
This operation increments the clock values of all configurations stored in $\Dt[i]$ by $r$, 
but may output a set of configurations whose clock values ceased to fit in the interval $J_i$. 
Recall that this set is organised as a 
doubly linked list of configurations, sorted by 
non-decreasing clock values; call this list $S_i$. 
Now, we need to insert each configuration $(q,t)$ that appears on those lists 
into the appropriate data structure $\Dt[j]$, where $j$ is such that $t\in J_j$. 
However, we have to be careful about the order of insertions:
we process the lists $S_{2k+1},S_{2k},\ldots,S_0$ in this precise order, 
and each list $S_i$ is processed from the end, that is, following the non-increasing 
order of clock values.
When processing a configuration $(q,t)$ from the list $S_i$, we find the index $j>i$ such that $t\in J_j$ and apply the method $\ins(q,t)$ on the structure $\Dt[j]$.
In this way the condition required by the $\ins$ method --- that $t\leq t'$ 
for every configuration $(q',t')$ currently stored in $\Dt[j]$ --- is satisfied.
It is also easy to see that Invariant~(I\ref{i:partition}) is preserved after the update.

\subparagraph{Update by a letter.} Suppose the next symbol read from the stream is 
a letter $a\in \Sigma$. We update the outer data structure as follows.
First, we apply method $\updLett(a)$ to each data structure $\Dt[i]$. 
This operation applies all possible transitions on letter $a$ to all configurations stored in $\Dt[i]$, 
and outputs a list of configurations $R_i$ where the clock got reset. %, organized as a doubly linked list. 
Note that all these configurations have clock value $0$, hence the length of the
list $R_i$ is at most $|Q|$, which is a constant. It now suffices to insert all the configurations $(q,0)$ appearing on all the lists $R_i$ to the data structure $\Dt[0]$ using method $\ins(q,0)$.
We may do this in any order, as the condition required by the $\ins$ method is trivially satisfied.
Again, Invariant~(I\ref{i:partition}) is clearly preserved after the update.

\smallskip

This concludes the implementation of the outer data structure. While the correctness is clear from the description, we are left with arguing that the time complexity is as promised.

Since $|Q|$ and $k$ are considered constants, it can be easily argued using Lemma~\ref{lem:inner} that under the implementation presented above, each of the following operations takes constant time: 
initialisation, a query about the acceptance, and an update by a letter.
As for an update by a time span $r\in \realplus$, 
by Lemma~\ref{lem:inner} the complexity of such an update is $\Oh(\sum_{i=0}^{2k+1} |S_i|)$, where $S_0,\ldots,S_{2k+1}$ are the sets returned by the applications of method $\updTime(r)$ to data structures $\Dt[0],\ldots,\Dt[2k+1]$,
respectively.
We need to argue that the amortised time complexity of all these updates is constant.

Consider the following definition: a clock value $t\in \real_{\geq 0}$ is {\em{active}} at step $n$ if $\Kt_n$ contains a configuration with clock value $t$.
Observe that upon an update by a time span $r\in \realplus$, the set of active clock values simply gets shifted by $r$, while upon an update by a letter $a\in \Sigma$ 
it stays the same, except that possibly clock value $0$ becomes active in addition. Since at step~$0$ the only active clock value is $0$, we conclude that for every $n\in \bbN$, 
at most $n+1$ active clock values may have appeared until step $n$.
Now observe that since $|Q|$ is considered a constant and there may be at most $|Q|$ different active configurations with the same active clock value, 
the complexity of each update by a time span is proportional to the number of active clock values that change the interval $J_i$ to which they belong, 
where we imagine that each active clock value is shifted by the time span.
As every active clock value can change its interval at most $2k+1=\Oh(1)$ times, and the total number of active values that appear until step $n$ is at most $n+1$,
we conclude that the total time spent on updates by time spans throughout the first $n$ steps is $\Oh(n)$. This means that the amortised time complexity is $\Oh(1)$.

Finally, note that in the case of discrete streams each set $S_i$ consists of configurations with the same clock value, hence $|S_i|\leq |Q|=\Oh(1)$ for all $i\in \{0,\ldots,2k+1\}$.
Consequently, in this case the complexity of an update by a time span is also constant, without any amortisation.

%\smallskip

This finishes the proof of Theorem~\ref{th:one_clock}, assuming Lemma~\ref{lem:inner}. We prove the latter next.

% moved to macros.tex
%\newcommand{\wLt}{\widehat{\Lt}}
%\newcommand{\rt}{\mathsf{root}}
%\newcommand{\ft}{\lambda}
%\newcommand{\clk}{\mathtt{y}}
%\newcommand{\lst}{\mathtt{l}}
%\newcommand{\lstr}{\mathtt{r}} % was \mathtt{lr}
%\newcommand{\lret}{\mathtt{lret}}
%\newcommand{\frs}{\mathtt{f}}

\subsection{Inner data structure}\label{sec:inner}%: implementation}

We now describe in detail
%present an implementation of 
the inner data structure $\Dt[J]$
%that is, 
and 
%we 
prove Lemma~\ref{lem:inner}.
Let us fix an interval $J\in \{J_0,\ldots,J_{2k+1}\}$.
We will denote by $L$ the set of configurations currently stored by the inner data structure $\Dt[J]$.
It is convenient to represent $\Lt$ by a function $\ft\colon \real_{\geq 0}\to 2^Q$ 
defined by
\<
  \ft(t)=\{\,q\in Q\ \colon\ (q,t)\in \Lt\}.
\>
%Note that function $\ft$ uniquely defines the set $\Lt$, and vice versa.
We let $\wLt$ be the set of all clock values that are {\em{active}} in $\Lt$, that is, $\wLt$ comprises all $t\in \real_{\geq 0}$ such that $\ft(t)\neq \emptyset$.
Recall that we assume that $\wLt\subseteq J$.

\paragraph*{Description of the structure.}
In short, the data structure $\Dt[J]$ consists of three elements:
\begin{itemize}
 \item The {\em{clock}}, denoted $\clk$, is a non-negative real that stores the total time elapsed so far.
 \item The {\em{list}}, denoted $\lst$, stores the set of active clock values $\wLt$.
 \item The {\em{forest}}, denoted $\frs$, is built on top of the elements of $\lst$ and describes the function $\ft$.
\end{itemize}

We describe the list and the forest in more details (the reader can refer to Figure \ref{fig:datastructure}).

\subparagraph{The list.} The list $\lst$ encodes the clock values present in $\wLt$, 
sorted in the increasing order and organised into a doubly linked list.
Each node $\lnode$ on $\lst$ is a record consisting of:
\begin{itemize}
 \item $\elnext(\lnode)$: a pointer to the next node on the list;
 \item $\elprev(\lnode)$: a pointer to the previous node on the list; and
 \item $\eltime(\lnode)\in\bbR$: the {\em{timestamp}} of the node.
\end{itemize}
As usual, the data structure stores $\lst$ by maintaining pointers to the first and last node.

The clock value represented by a node $\lnode$ on $\lst$ 
(represented by a square in Figure \ref{fig:datastructure}) 
is equal to $\aclock(\lnode) = \clk - \eltime(\lnode)$;
this will always be a non-negative real.
Thus, intuitively, the timestamp is essentially the total elapsed time recorded 
since the last reset of the clock.
Note that this implementation allows for a simultaneous increment of $\aclock(\lnode)$ for all nodes $\lnode$ on $\lst$ in constant time: it suffices to simply increment $\clk$.
%If $\lnode$ is the first node on $\lst$, then we put $\elprev(\lnode)=\bot$, 
%and similarly $\elnext(\lnode)=\bot$ for the last node $\lnode$ of $\lst$.

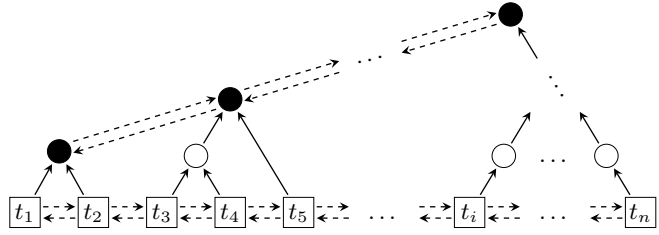
\begin{figure}
\vspace{-3mm}
\centering
\begin{tikzpicture}[xscale=0.9,yscale=0.75]
\clip (-0.3,-0.4) rectangle (9.3,3.9);

\draw (0,0) node [cell] (t1) {\small $t_1$};
\draw (1,0) node [cell] (t2) {\small $t_2$};
\draw (2,0) node [cell] (t3) {\small $t_3$};
\draw (3,0) node [cell] (t4) {\small $t_4$};
\draw (4,0) node [cell] (t5) {\small $t_5$};
\draw (5,0) node [cellh] (t6) {};
\draw (5.25,0) node [cellh] {\small $\dots$};
\draw (5.5,0) node [cellh] (t7) {};
\draw (6.5,0) node [cell] (t8) {\small $t_i$};
\draw (7.5,0) node [cellh] (t9) {};
\draw (7.75,0) node [cellh] {\small $\dots$};
\draw (8,0) node [cellh] (t10) {};
\draw (9,0) node [cell] (t11) {\small $t_n$};

\draw [transform canvas={yshift=0.75mm}] (t1) edge [dotted, arrow] (t2);
\draw [transform canvas={yshift=-0.75mm}] (t2) edge [dotted, arrow] (t1);
\draw [transform canvas={yshift=0.75mm}] (t2) edge [dotted, arrow] (t3);
\draw [transform canvas={yshift=-0.75mm}] (t3) edge [dotted, arrow] (t2);
\draw [transform canvas={yshift=0.75mm}] (t3) edge [dotted, arrow] (t4);
\draw [transform canvas={yshift=-0.75mm}] (t4) edge [dotted, arrow] (t3);
\draw [transform canvas={yshift=0.75mm}] (t4) edge [dotted, arrow] (t5);
\draw [transform canvas={yshift=-0.75mm}] (t5) edge [dotted, arrow] (t4);
\draw [transform canvas={yshift=0.75mm}] (t5) edge [dotted, arrow] (t6);
\draw [transform canvas={yshift=-0.75mm}] (t6) edge [dotted, arrow] (t5);
\draw [transform canvas={yshift=0.75mm}] (t7) edge [dotted, arrow] (t8);
\draw [transform canvas={yshift=-0.75mm}] (t8) edge [dotted, arrow] (t7);
\draw [transform canvas={yshift=0.75mm}] (t8) edge [dotted, arrow] (t9);
\draw [transform canvas={yshift=-0.75mm}] (t9) edge [dotted, arrow] (t8);
\draw [transform canvas={yshift=0.75mm}] (t10) edge [dotted, arrow] (t11);
\draw [transform canvas={yshift=-0.75mm}] (t11) edge [dotted, arrow] (t10);

\draw (0.5,1.08) node [circlef] (r1) {};
  \draw (2.5,1) node [circle] (f2) {};
\draw (3,2) node [circlef] (r2) {};
  \draw (4.8,2.65) node [circleh] (r3) {};
\draw (5.05,2.8) node [circleh, rotate=14] {\small $\dots$};
  \draw (5.3,2.85) node [circleh] (r4) {};
  \draw (7,1) node [circle] (h5) {};
  \draw (7.75,1) node [circleh] (g5) {\small $\dots$};
  \draw (8.5,1) node [circle] (f5) {};
  \draw (8,2) node [circleh] (f5') {};
  \draw (7.5,2) node [circleh] (h5') {};
  \draw (7.85,2.25) node [circleh, rotate=-58] {\small $\dots$};
  \draw (7.6,2.5) node [circleh] (f5'') {};
\draw (7.1,3.5) node [circlef] (r5) {};

\draw (t1) edge [arrow] (r1);
\draw (t2) edge [arrow] (r1);
\draw (t3) edge [arrow] (f2);
\draw (t4) edge [arrow] (f2);
\draw (f2) edge [arrow] (r2);
\draw (t5) edge [arrow] (r2);
\draw (t8) edge [arrow] (h5);
\draw (t11) edge [arrow] (f5);
\draw (f5) edge [arrow] (f5');
\draw (h5) edge [arrow] (h5');
\draw (f5'') edge [arrow] (r5);

\draw [transform canvas={yshift=0.75mm}] (r1) edge [dotted, arrow] (r2);
\draw [transform canvas={yshift=-0.75mm}] (r2) edge [dotted, arrow] (r1);
\draw [transform canvas={yshift=0.75mm}] (r2) edge [dotted, arrow] (r3);
\draw [transform canvas={yshift=-0.75mm}] (r3) edge [dotted, arrow] (r2);
\draw [transform canvas={yshift=0.75mm}] (r4) edge [dotted, arrow] (r5);
\draw [transform canvas={yshift=-0.75mm}] (r5) edge [dotted, arrow] (r4);
\end{tikzpicture}

\caption{The inner data structure.}\label{fig:datastructure}
\vspace{-3mm}
\end{figure}

\subparagraph{The forest.} The forest $\frs$ represents the assignment 
that maps elements $t\in \wLt$, encoded in $\lst$, to respective sets of control states $\ft(t)$.
It has a standard form of a rooted forest, where every node may have arbitrarily many children, 
and these children are unordered.
Every node $\fnode$ of $\frs$ (represented by a circle in Figure \ref{fig:datastructure}) 
is a record containing:
\begin{itemize}
 \item $\parent(\fnode)$: a pointer to the parent of $\fnode$; and
 \item $\elchildren(\fnode)$: an integer equal to the number of children of $\fnode$.
\end{itemize}
The leaves of the forest will always coincide with the nodes on the list $\lst$.
Thus, we may simply augment the records stored for the nodes on $\lst$ by adding 
the $\parent(\cdot)$ pointer, and treat them as nodes of the forest $\frs$ at the same time.
The counter $\elchildren(\cdot)$ would always be equal to $0$ for those nodes, so we may omit it.

The {\em{roots}} of the forest (represented by the black circles in Figure \ref{fig:datastructure}) 
are the nodes $\froot$ with no parent, i.e.~$\parent(\froot)=\bot$.
We will maintain the invariant that no root is a leaf in $\frs$, that is, every root has at least one child.
In the data structure we will maintain a doubly linked list containing all the roots of $\frs$.
This list will be denoted 
$\lstr$, and again it will be stored by pointers to its first and last element.
Thus, the records of the roots of $\frs$ are augmented by $\elnext(\cdot)$ and $\elprev(\cdot)$ pointers describing the structure of $\lstr$, with the usual meaning.
In addition to this, every root $\froot$ of $\frs$ carries two additional values:
\begin{itemize}
 \item $\elstates(\froot)\subseteq Q$: a non-empty subset of control states for which $\froot$ is responsible; and
 \item $\elrank(\froot)$: an integer from the set $\{1,2,\ldots,2^{|Q|}\}$.
%  \gabriele{it seems that this rank is not necessary, as it can be derived from, e.g. the ordering of the roots
%            (assuming roots are reordered properly), but I agree it is convenient}
\end{itemize}
We will maintain two invariants about these values.
First, the sets $\elstates(\froot)$ and the ranks $\elrank(\froot)$ should be pairwise different for distinct roots $\froot$ of $\frs$.
Note that this means that $\frs$ always has at most $2^{|Q|}-1=\Oh(1)$ roots.
Second, for every root $\froot$, the {\em{tree rooted at $\froot$}} --- which is the tree consisting of $\froot$ and all its descendants in $\frs$ --- has depth at most $\elrank(\froot)$.
Here, the {\em{depth}} of a forest is the maximum number of 
% \gabriele{replaced nodes by edges in the def. of depth, since otherwise depth would always be $\ge2$}
edges
%nodes 
on a path from a leaf to a root, minus~$1$.
Note that this implies that the depth of the forest $\frs$ is bounded by $2^{|Q|}=\Oh(1)$.

The function $\ft$ is then represented as follows. For every node $\lnode$ on $\lst$, let $\rt(\lnode)$ be the root of the tree of $\frs$ that contains $\lnode$.
Then denoting $t=\aclock(\lnode)$, we have $\ft(t)=\elstates(\rt(\lnode))$.
Note that the invariant stated above imply that from every leaf $\lnode$ of $\frs$, $\rt(\lnode)$ can be computed from $\lnode$ by following the $\parent(\cdot)$ pointer at most $2^{|Q|}=\Oh(1)$ times.
Hence, given $t\in \wLt$ together with a node $\lnode$ on $\lst$ satisfying $t=\aclock(\lnode)$, we can compute $\ft(t)$ in $\Oh(1)$ time.

\paragraph*{Invariants.} For convenience, we gather all the invariants 
maintained by the inner data structure which we mentioned before:
\begin{enumerate}[{I}1.]\setcounter{enumi}{1}
 \item\label{j:list-elts} For each node $\lnode$ on $\lst$, the value $\aclock(\lnode) = \mathtt{y} - \eltime(\lnode)$ belongs to $J$.
 \item\label{j:list-sort} The nodes on $\lst$ are sorted by increasing clock values, 
 or equally by
%  \gabriele{here I guess the order on timestamps should be decreasing.}
 decreasing
% increasing 
timestamps. 
That is, 
$\eltime(\lnode)>\eltime(\elnext(\lnode))$ 
%$\eltime(\elnext(\lnode))>\eltime(\lnode)$
for every non-last node $\lnode$ on $\lst$.
 \item\label{j:non-leaves} Every root of $\frs$ has at least one child, and the leaves of $\frs$ are exactly all the nodes on $\lst$.
 \item\label{j:roots} The roots of $\frs$ carry pairwise different, non-empty sets of control states, and they have pairwise different ranks. 
                      Moreover, all the ranks belong to the set $\{1,2,\ldots,2^{|Q|}\}$.
 \item\label{j:depth} For every root $\froot$ of $\frs$, the depth of the tree rooted at $\froot$ is at most $\rank(\froot)$.
\end{enumerate}

\paragraph*{Implementation.} We now show how to implement the methods
$\inz(J)$, 
$\acc()$, $\ins(q,t)$, $\updTime(r)$, and $\updLett(a)$ in the data structure.
Recall that all these methods should work in constant time, with the exception of $\updTime(r)$ which is allowed to work in time linear in its output.
The description of each method is supplied by a running time analysis and an argumentation of the correctness, which includes a discussion on why the invariants stated above are maintained.

\subparagraph{Removing nodes.} Before we proceed to the description of the required methods, we briefly discuss an auxiliary procedure of removing a node from the list $\lst$ and from the forest $\frs$,
as this procedure will be used several times.
Suppose we are given a node $\lnode$ on the list $\lst$ and we would like to remove it, which corresponds to removing from $\Lt$ all configurations $(q,t)$ where $t=\aclock(\lnode)$ and $q\in \ft(t)$.
We can remove $\lnode$ from $\lst$ in the usual way. Then we remove $\lnode$ from $\frs$ as follows. First, we decrement the counter of children in the parent of $\lnode$. 
If this counter stays positive then there is nothing more to do. 
Otherwise, we need to remove the parent of $\lnode$ as well, and accordingly decrement the counter of children in the grandparent of $\lnode$.
This can again trigger removal of the grandparent and so on. If eventually we need to remove a root of $\frs$, we also remove it from the list $\lstr$ in the usual way. 
Note that since by Invariants~(I\ref{j:roots}) and~(I\ref{j:depth}), the depth of $\frs$ is bounded by a constant, the total number of removals is bounded by a constant as well, 
and the whole procedure can be performed in constant time. It is straightforward to verify that all the invariants are maintained.

\subparagraph{Initialization.} The $\inz(J)$ method 
stores the interval $J$, that defines the range of clock values that could be
represented in the data structure. It also sets $\clk=0$ and initialises 
$\lst$ and $\lstr$ as empty lists. The correctness and the running time are clear. 

\subparagraph{Acceptance query.} The $\acc()$ method is implemented as follows. 
We iterate through the list $\lstr$ to check whether there exists a root $\froot$ of $\frs$ such that $\elstates(\frs)$ contains any accepting state, say $q$.
If this is the case, then by Invariant~(I\ref{j:non-leaves}) there is a node $\lnode$ on $\lst$ satisfying $\rt(\lnode)=\froot$, hence $(q,t)$ is an accepting configuration that belongs to~$\Lt$, 
where $t=\aclock(\lnode)$. So we may return a positive answer from the query. Otherwise, all configurations in $\Lt$ have non-accepting states, and we may return a negative answer.
Note that since by Invariant~(I\ref{j:roots}) the list $\lstr$ has bounded length, the above procedure works in constant time.

\subparagraph{Insertion.} We now implement the method $\ins(q,t)$, where $(q,t)$ is a configuration. Recall that when this method is executed, we have a promise that $t\in J$ and $t\leq t'$ for all 
configurations $(q',t')$ that are currently present in $\Dt[J]$.

Let $\lnode$ be the first node on the list $\lst$ and let $t'=\aclock(\lnode)$. 
By the promise, we have $t\leq t'$. We distinguish two cases: either $t<t'$ or $t=t'$. The former case also encompasses the corner situation when $\lst$ is empty.

When $t<t'$ or $\lst$ is empty, the new configuration $(q,t)$ gives rise to a new active clock value $t$. Therefore, we create a new list node $\lnode_0$ and insert it at the front of the list $\lst$.
We set the timestamp as $\eltime(\lnode_0)=\clk-t$, so that the node correctly represents the clock value $t$. 
It is clear that Invariants~(I\ref{j:list-elts}) and~(I\ref{j:list-sort}) are thus satisfied.

Next, we need to insert the new node $\lnode_0$ to the forest $\frs$. We iterate through the list $\lstr$ in search for a root $\froot$ that satisfies $\elstates(\froot)=\{q\}$.
In case there is one, we simply set $\parent(\lnode_0)=\froot$ and increment $\elchildren(\froot)$. 
Otherwise, we construct a new root $\froot_0$ with $\elstates(\froot_0)=\{q\}$ and $\elchildren(\froot_0)=1$, insert it at the front of the list $\lstr$, and set $\parent(\lnode_0)=\froot_0$.
To determine the rank of $\froot_0$, we find the smallest integer $k\in \{1,\ldots,2^{|Q|}\}$ that is {\em{not}} used as the rank of any other root of $\frs$. Observe that, by Invariant~(I\ref{j:roots}), the forest $\frs$ 
has at most $2^{|Q|}-1$ roots, so there is always such a number $k$, and it can be found in constant time by inspecting the list $\lstr$. We then set $\elrank(\froot_0)=k$.
It is clear from the description that this operation can be performed in constant time, and that Invariants~(I\ref{j:non-leaves}),~(I\ref{j:roots}), and~(I\ref{j:depth}) are maintained. 
For the last one, observe that the new leaf $\lnode_0$ is attached directly under a root of $\frs$, so no tree in $\frs$ existing before the insertion could have increased its depth.

We are left with the case when $t=t'$. We first compute the set $X$ equal to $\ft(t)$ before the insertion: it suffices to find $\rt(\lnode)$ in constant time 
and read $X=\elstates(\rt(\lnode))$.
If $q\in X$ then the configuration $(q,t)$ is already present in $\Lt$, so there is nothing to do. Otherwise, we need to update the data structure so that $\ft(t)$ is equal to $X\cup \{q\}$ instead of $X$. 
Consequently, we remove the node $\lnode$ from $\lst$ and from $\frs$, 
%(recall that the operation of {\em{removing}} a node was described earlier), 
using the operation described earlier,
and we insert a new node $\lnode'$ at the front of $\lst$,
with the same timestamp as that of $\lnode$. Thus, $\aclock(\lnode')=t$. We next insert the new node $\lnode'$ to the forest $\frs$ using the same procedure as described in the previous paragraph, 
but applied to the state set $X\cup \{q\}$ instead of $\{q\}$. Again, it is clear that these operations can be performed in constant time, and the same argumentation shows that all the invariants are maintained.

\subparagraph{Update by a time span.} Next, we implement the method $\updTime(r)$, where $r\in \realplus$.

First, we increment $\clk$ by $r$. Thus, for every node $\lnode$ in the list $\lst$ the value $\aclock(\lnode)$ got incremented by $r$.
However, the Invariant~(I\ref{j:list-elts}) may have ceased to hold, for some active clock values could have been shifted outside of the interval $J$.
The configurations with these clock values should be removed from the data structure and their list should be returned as the output of the method.

We extract these configurations as follows. Construct an initially empty list of configuration $\lret$, on which we shall build the output.
Iterate through the list $\lst$, starting from its back. For each consecutive node $\lnode$, compute $t=\aclock(\lnode)$. 
If $t\in J$, then break the iteration and return $\lret$, as there are no more configurations to remove. 
Otherwise, find $\rt(\lnode)$ in constant time, read $\ft(t)=\elstates(\rt(\lnode))$, and add at the front of $\lret$ all configurations $(q,t)$ for $q\in \ft(t)$, in any order.
Then remove $\lnode$ from the list $\lst$ and from the forest $\frs$, and proceed to the previous node in $\lst$ (if there is none, finish the iteration).

By Invariant~(I\ref{j:list-sort}), it is clear that in this way we remove from $\Dt[J]$ exactly all the configurations whose clock values got shifted outside of $J$, hence Invariants~(I\ref{j:list-elts}) and~(I\ref{j:list-sort})
are maintained. As the forest structure was influenced only by removals, Invariants~(I\ref{j:non-leaves}),~(I\ref{j:roots}), and~(I\ref{j:depth}) are maintained as well.
Also note that the configurations on the output list $\lret$ are ordered by non-decreasing clock values, as was required.

As for the time complexity, recall that by Invariant~(I\ref{j:roots}), for every removed node $\lnode$ the set $\elstates(\rt(\lnode))$ is non-empty and has size at most $|Q|$.
Hence, with every removed node $\lnode$ we add to $\lret$ between $1$ and $|Q|=\Oh(1)$ new configurations. As the time complexity of the procedure
is bounded linearly in the number of nodes that we remove from $\lst$, it is also
bounded linearly in the number of configurations that appear in the output list $\lret$.

\subparagraph{Update by a letter.} We proceed to the method $\updLett(a)$, where $a\in \Sigma$.
As argued before, every clock condition appearing in $\Aa$ is either true for all clock values in $J$, 
or false for all clock values in $J$.
For every subset of states $X\subseteq Q$, let $\Phi(X)$ be the set of all states $q$ such that there is a transition $(p,a,q,\gamma,\emptyset)$ in $E$ for some $p\in X$ and clock condition $\gamma$ that is true in $J$.
In other words, $\Phi(X)$ comprises states reachable from the states of $X$ by non-resetting transitions over $a$ that are available for clock values in $J$.
We define $\Psi(X)$ in a similar way, but for resetting transitions over $a$ 
that are available for clock values in $J$.

% \gabriele{replaced $Z$ by $X$, which can be reused w.r.t. previous meaning}
First, we compute the output of the method, which should be simply $\{(q,0)\ \colon\ q\in \Psi(X)\}$, 
where $X$ is the set of all states appearing in the configurations of $\Lt$.
Observe that, by Invariant~(I\ref{j:non-leaves}), $X$ can be computed in constant time by iterating through the list $\lstr$ and computing the union of sets $\elstates(\froot)$ for roots $\froot$ appearing on it.
Thus, the output of the method can be computed in constant time.

Second, we need to update the values of function $\ft$ by applying all possible non-resetting transitions over $a$.
This can be done by iterating through the list $\lstr$ and, for each root $\froot$ appearing on it, substituting $\elstates(\froot)$ with $\Phi(\elstates(\froot))$.
Note that since we assumed that for every state $q$, some transition over $a$ is always available at $q$, it follows that $\Phi$ maps non-empty sets of states to non-empty sets of states. 
Hence, after this substitution the roots of $\frs$ will still be assigned non-empty sets of states.
However, Invariant~(I\ref{j:roots}) may cease to hold, as some roots may now be assigned the same set of states.

We fix this as follows. For every root $\froot$ of $\frs$, inspect the list $\lstr$ and find the root $\froot'$ that has the largest rank among those satisfying $\elstates(\froot)=\elstates(\froot')$.
If $\froot=\froot'$, then do nothing. Otherwise, turn $\froot$ into a non-root node of $\frs$, remove it from the list $\lstr$, set $\parent(\froot)=\froot'$, and increment $\elchildren(\froot')$ by one.
Note that after applying this modification, the function $\ft$ stored in the data structure stays the same, while Invariant~(I\ref{j:roots}) becomes satisfied.

As for the other invariants, the satisfaction of Invariants~(I\ref{j:list-elts}),~(I\ref{j:list-sort}), and~(I\ref{j:non-leaves}) after the update is clear.
However, we need to be careful about Invariant~(I\ref{j:depth}), as we might have substantially modified the structure of the forest $\frs$.
Observe that each modification of $\frs$ that we applied boils down to attaching a tree with a root of some rank $i$ as a child of a tree with a root of some rank $j>i$.
By Invariant~(I\ref{j:depth}), the former tree has depth at most $i$, which is 
bounded from above by $j-1$. Thus, after the attachment, the depth of the latter tree cannot become larger than $j$.
We conclude that Invariant~(I\ref{j:depth}) is maintained as well.

Finally, note that since the number of roots of $\frs$ is always bounded by a constant, all the operations described above can be performed in constant time.

\smallskip

We have implemented all the required methods within the claimed running time bounds. This concludes the proofs of Lemma~\ref{lem:inner} and of Theorem~\ref{th:one_clock}.

%\section{Lower bound}
%\label{sec:lower}
%\input{lower}

\section{Conclusion}
\label{sec:conclusion}
We have considered a monitoring problem for streams processed 
by a timed automaton, which consists of readily answering
the membership query ``\emph{Does the time automaton accept
the current prefix of the stream?}''.
% , while parsing
% longer and longer prefixes of the input.
We have designed a suitable data structure that solves
the monitoring problem for a one-clock timed automaton
in amortised constant time
% , that is, in time that is 
% on average constant per query and per consumed element 
% of the stream,
assuming that the timed automaton is
fixed.
%
%The original motivation was to study how complex it is to simulate the 
%runs of an automaton on longer and longer prefixes of an input word.
%Of course, this type of question has a simple answer on finite state 
%automata, since for any fixed finite automaton, only boundedly many 
%states can be reached on a given input prefix. For timed automata, 
%however, the situation is different, since the set of reachable 
%configurations is not a priori bounded.
%The goal was thus to find a suitable data structure
%for representing sets of runs of a timed automaton that 
%could be updated and queried efficiently, possibly 
%in time that is approximately constant per input letter, 
%and independently on the actual constants that appear
%in the input and in the transition constraints of the 
%timed automaton.
%We answered positively this question, showing that
%there is a data structure for monitoring runs
%of a one-clock timed automaton that could be updated 
%in amortized constant time on consuming a letter from 
%the input stream. We also provide some
%lower bounds, showing that under mild complexity
%assumptions ({3SUM} Conjecture), there is no data
%structure that supports monitoring of a two-clock
%timed automaton, enhanced with additive constraints,
%in amortised time that is strongly sub-linear.
 
We leave as an open question whether our complexity
result can be strengthened, for instance, by devising
amortised constant time monitoring algorithms 
for timed automata with multiple clocks.
Concerning the latter question, we proved that,
assuming the {\sc{3SUM}} Conjecture,
this is not possible for timed automata enhanced 
with additive clock constraints.

\bibliography{biblio}

\newpage
\appendix

\section{Other application examples for the monitoring problem}
\label{appendix:examples}
\begin{example}\label{ex:within}
Here we consider a scenario from \emph{complex event processing} (CEP), with a specification
language called CEL and defined by the following grammar \cite{grez2019formal}:
\[
  \varphi ~:=~ a ~\mid~ \varphi ; \varphi ~\mid~ \varphi \within t
\]
where $a \in \Sigma$ and $t \in \bbN$.
A word $w=a_1 a_2 \ldots a_n \in \Sigma^*$ \emph{matches} an expression $\varphi$ from the above grammar, 
denoted $w \vDash \varphi$, if one of the following cases holds:
\begin{itemize}
	\item $\varphi = a_n$,
	\item $\varphi = \varphi_1 ; \varphi_2$, $w = w_1 \cdot w_2$, 
	      $w_1 \vDash \varphi_1$ and $w_2 \vDash \varphi_2$,
	\item $\varphi = \varphi' \within t$ and $a_m\ldots a_n \vDash \varphi'$, where $m = \max\{1,n-t\}$. 
\end{itemize} 
%Intuitively, operator $;$ allows the left and right patterns to appear arbitrarily away from one another, while $\within$ restricts the pattern to appear within a substring of bounded length.
%Moreover, note that the last symbol $a_n$ must always be part of the pattern.

Given a word $w=a_1 a_2 \dots$ and an expression $\varphi$, we would like to read $w$ sequentially,
as in a stream, and decide, at each position $n=1,2,\dots$, whether the prefix 
$w_n = a_1 \dots a_n$ matches a fixed expression $\varphi$.
One can reduce this latter problem to our monitoring problem for timed automata, 
by using a discrete timed word $\ww = 1 a_1 1 a_2 1 \dots$ as before and by 
translating the expression $\varphi$ into an appropriate timed automaton.
We omit the straightforward details of the translation of a CEL expression to an equivalent timed automaton,
and we only remark that every occurrence of the $\within$ operator in an expression corresponds
to a condition on a specific clock in the equivalent timed automaton. This means that, in general,
the translation may require a timed automaton with multiple clocks. However, there are simple cases
(which we do not characterise here) where,
even in the presence of nested $\within$ operators, one can construct an equivalent timed automaton
with a single clock.
For example, 
consider the expression $\varphi = ((a ; b) \within 4); c \within 10$, which describes
a sequence containing three (possibly not contiguous) events $a,b,c$, with $a$ and $b$ at distance 
at most $4$ and $a$ and $c$ at distance at most $10$.
Figure~\ref{fig:within} shows a single-clock timed automaton that is equivalent to $\varphi$,
in the sense that it accepts a timed word of the form $1a_11a_21\ldots1a_n$ 
if and only if $a_1a_2 \ldots a_n \vDash \varphi$.
In this case one can validate any input stream against the expression $\varphi$ 
in time that is constant per input letter, by simply reducing to our monitoring 
problem for single-clock timed automata and discrete timed words.
\end{example}

\begin{figure}
\centering
\begin{tikzpicture}[scale=0.9]
\begin{scope}
  \draw (0,1.5) node [right] {$\varphi = ((a ; b) \within 4); c \within 10$};
  
  \draw [dotted] (0,0) -- (6,0);

  \draw (1,0) -- (1,0.2);
  \draw (3,0) -- (3,0.2);
  \draw (5,0) -- (5,0.2);
  
  \draw (1,0.5) node [baseline] (a) {\small $a$};
  \draw (3,0.5) node [baseline] (b) {\small $b$};
  \draw (5,0.5) node [baseline] (c) {\small $c$};
  
  \draw (1,-0.2) -- (1,-0.4) -- node [below=-0.05] {\small $\le 4$} (3,-0.4) -- (3,-0.2);
  \draw (1,-0.6) -- (1,-0.8) -- node [below=-0.05] {\small $\le 10$} (5,-0.8) -- (5,-0.2);
\end{scope}
\begin{scope}[xshift=9cm]
  \draw (-1.5,0.2) node {$\hat\Aa:$};
  
  \draw (0,0) node [dot] (p) {};
  \draw (2,0) node [dot] (q) {};
  \draw (4,0) node [dot] (r) {};
  \draw (6,0) node [dot] (s) {};
  
  \draw [arrow] (-0.5,0) to (p);
  \draw [arrow] (s) to (6.5,0);
  \draw [arrow] (p) to [out=45, in=135, looseness=35] node [black, above=-0.05] 
        {\small $*$} (p);
  \draw [arrow] (p) to node [black, below=-0.05] {\small $a\mid\text{reset } \mathtt{x}$} (q);
  \draw [arrow] (q) to node [black, below=-0.05] {\small $b,\mathtt{x}\le 4$} (r);
  \draw [arrow] (r) to node [black, below=-0.05] {\small $c,\mathtt{x}\le 10$} (s);
\end{scope}
\end{tikzpicture}
\caption{Translation of a CEL expression into an equivalent single-clock timed automaton.}\label{fig:within}
\end{figure}
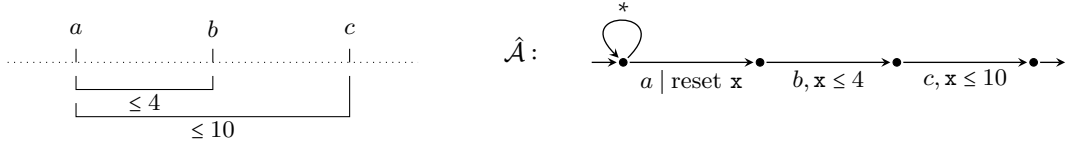

\begin{example}\label{ex:frobenius}
Consider a list $k_1,\ldots,k_n \in \nat$ of numbers whose greatest common divisor is~$1$. 
The Frobenius coin problem deals the following question: what is the maximum
integer number $h$ that cannot be expressed as a sum of multiples of $k_i$'s?
We consider a dynamic variant of this problem, where we are given $k_1,\dots,k_n\in\nat$ on input, and then
for consecutive numbers $h=1,2,3,\dots$ we should answer whether $h$ can be expressed as a sum of multiples of $k_i$'s.
The time complexity should be constant per each choice of $h$, assuming that the numbers $h=1,2,3,\dots$ 
arrive in this precise ordering.

The idea is to construct a timed automaton $\Aa$ that consumes longer and longer 
prefixes $(1a)^h$ of a potentially infinite stream, while 
checking whether the total elapsed time $h$ can be written as 
$\sum_{i=1}^n k_i\cdot h_i$ for some $h_1,\dots,h_n\in\bbN$, 
in which case the prefix is accepted. 
%The monitoring algorithm 
%stops as soon as $\cA$ rejects, thus solving the Frobenius coin problem. 
% \gabriele{is it now correct with the new final state?}
Formally, we let $\Aa=(\Sigma,Q,X,I,E,F)$, where 
$\Sigma = \{a\}$, $Q = \{p,q\}$, $X=\{\mathtt{x}\}$, $I = \{p\}$, $F = \{q\}$, 
and $E$ contains the following transitions:
$(p,a,\true,p,\emptyset)$, $(p,a,\varphi,q,\{\mathtt{x}\})$, and $(q,a,\true,p,\emptyset)$, 
with $\varphi=(\mathtt{x}=k_1)\vee\dots\vee(\mathtt{x}=k_n)$.
Intuitively, the first transition lets some arbitrary amount of time to pass, 
until the second transition is enabled, which happens when the clock value 
is any of the constants $k_1,\dots,k_n$. Every time the second transition is
triggered, the clock is reset and either the word terminates and is accepted,
or a subsequent transition brings the control back to the initial state.
It is easy to see that an input word of the form $(1a)^h$ is accepted by 
this timed automaton if and only if $h=\sum_{i=1}^n k_i\cdot h_i$ 
for some $h_1,\dots,h_n\in\bbN$. 
\end{example}

\section{3SUM Conjecture and lower bound for the monitoring problem}
In this section, we prove a complexity lower bound for a variant of our 
monitoring problem.
Ideally, we would like to prove that there is a timed automaton with two clocks 
for which monitoring in amortised constant time is not possible. This would 
imply that our result (Theorem \ref{th:one_clock}) 
for monitoring a single-clock timed automaton cannot be generalised to the multiple-clock setting.
We are not able to establish optimality in this sense. 
We can however prove a result along the same line, by considering 
timed automata extended with additive constraints, that is, having 
clock conditions of the form $\big(\sum_{\mathtt{x}\in Z} \mathtt{x}\big) \sim c$.
Our lower bound is based on the {\sc{3SUM}} Conjecture, which we restate below for convenience.

\SumConjecture

Recall that in the {\sc{3SUM}} problem we are given a set $S$ of positive real numbers 
and the question is to determine whether there exist $a,b,c\in S$ satisfying $a+b=c$.
We remark that the original phrasing of the conjecture allows non-positive numbers on 
input and asks for $a,b,c\in S$ such that $a+b+c=0$.
It is easy to reduce this standard formulation to our setting, for example by replacing 
$S$ with $S'=\{3M+x \:\colon\: x\in S\}\cup \{6M-x \:\colon\: x\in S\}$, where $M$ is any 
real satisfying $M>|a|$ for all $a\in S$.

%It is easy to solve the 3SUM problem in time $\Oh(n^2)$, where $n=|S|$. 
%%: sort $S$ and for each pair $a,b\in S$, use binary search on sorted $S$ to check whether $a+b\in S$.
%For our lower bound, we will use the following hypothesis, stating that 
%the quadratic running time for 3SUM cannot be significantly improved:

%The {\sc{3SUM}} Conjecture is among the most popular hypotheses considered in computational geometry and in fine-grained complexity theory; see e.g. an overview in~\cite[Appendix~A]{AbboudWY18}.
%It was introduced by Gajentaan and Overmars~\cite{GajentaanO95,GajentaanO12} in a stronger form, which postulated the non-existence of {\em{sub-quadratic}} algorithms, that is, achieving running time $o(n^2)$.
%This formulation was refuted by Gr{\o}nlund and Pettie~\cite{GronlundP18}, who gave an algorithm for {\sc{3SUM}} with running time $\Oh(n^{2}/(\log n / \log \log n)^{2/3})$ in the real RAM model, 
%which can be improved to $\Oh(n^{2}(\log \log n)^2/\log n)$ when randomization is allowed. However, the existence of a strongly sub-quadratic algorithm is widely open and conjectured to be hard.

The {\sc{3SUM}} Conjecture has received significant attention in the recent years, as it was realised that it can be used as a base for tight complexity lower bounds 
for a variety of discrete graph problems, including questions about efficient dynamic data structures~\cite{AbboudW14,AlmanMW17,KopelowitzPP16,Patrascu10}.
In this setting, it is common to assume the integer formulation of the conjecture: there exists $d\in \N$ such that the {\sc{3SUM}} problem where the input numbers are integers
from the range $[-n^d,n^d]$ cannot be solved in strongly sub-quadratic time, assuming the word RAM model with words of bit length $\Oh(\log n)$. 
It is straightforward to verify that the construction we are going to present in this section can be turned into an analogous lower bound assuming the integer formulation of the {\sc{3SUM}} Conjecture.
For this, we would need to amend the formulation of the monitoring problem by assuming that the input stream is expected to have total length at most $N$, 
the clock constants and the time spans in the stream are integers of bit length at most $M$, 
% \gabriele{as before, changed $w$ into $M$}
and the data structure solving the monitoring problem should work in the word RAM model with words of bit length $\Oh(M+\log N)$.

%There exists $d\in \N$ such that, in the word RAM model with words of bit length $\Oh(\log n)$, 
%the {\sc{3SUM}} problem on $n$ integers from the range $\{1,\ldots,n^d\}$ 
%cannot be solved in strongly sub-quadratic time, that is, in time $\Oh(n^{2-\delta})$ 
%for any $\delta>0$.

\begin{comment}
We remark that the original phrasing of the conjecture allows non-positive 
integers on input, that is, assumes that the numbers are in the range 
$\{-n^d,\ldots,n^d\}$ and asks to find $a,b,c\in S$ such that $a+b+c=0$.
It is easy to reduce this standard formulation of {\sc{3SUM}} to our setting,
where only positive numbers are allowed, for instance by replacing $S$ 
with $S'=\{2M+a \:\colon\: a\in S\}\cup \{4M-a \:\colon\: a\in S\}$,
where $M$ is any integer satisfying $M>|a|$ for all $a\in S$.
From now on, we adopt the constant~$d$ postulated by the 
{\sc{3SUM}} Conjecture in the notation.
\end{comment}

We now prove Theorem~\ref{th:two_clocks}, restated below for convenience, which provides a lower bound for monitoring two-clock timed automata with additive constraints under the {\sc{3SUM}} Conjecture.

\TwoClocks*

Our approach is similar in spirit to the other lower bounds on dynamic problems, which we mentioned above~\cite{AbboudW14,AlmanMW17,KopelowitzPP16,Patrascu10}.
We first prove {\sc{3SUM}}-hardness of deciding acceptance by a timed automaton with additive constraints
in the static setting.
We then show that any data structure that supports monitoring
in amortised strongly sub-linear time would violate 
the {\sc{3SUM}}-hardness of the former static acceptance problem,
thus proving Theorem \ref{th:two_clocks}.

The postulated hardness of the static problem is captured by the following lemma.

\newcommand{\step}{\diamondsuit}
\newcommand{\phase}{\spadesuit}

\begin{lemma}\label{lem:beaver}
If the {\sc{3SUM}} Conjecture holds, then 
there is a two-clock timed automaton $\Aa$ with additive constraints
for which there is no algorithm that, given a timed word $w\in \big(\Sigma \uplus \realplus\big)^*$ 
as input, where $\Sigma$ is a two-letter alphabet, 
decides whether $\Aa$ accepts $w$ in time $\Oh(n^{2-\delta})$ for any $\delta>0$.
\end{lemma}
% There exists an \extended{} timed automaton $\Tt$ over a two-symbol alphabet $\Sigma$ and with two clocks such that, assuming the {\sc{3SUM}} Conjecture, 
% there is no algorithm that given a word $w\in (\Sigma \cup \{1,\ldots,2n^d\})^\star$ of length $n$ would decide whether $\Tt$ accepts $w$ in time $\Oh(n^{2-\delta})$, for any $\delta>0$.
%
%\gabriele[inline]{Maybe we can move the proof in the appendix?}
%
\begin{proof}
We construct a two-clock timed automaton $\Aa$ with additive constraints and 
an algorithm that given a set $S$ of $n$ positive reals, 
outputs a word $w\in \big(\Sigma \uplus \realplus\big)^*$ such that 
$w$ is accepted by $\Aa$ if and only if there are $a,b,c\in S$ satisfying $a+b=c$.
We find it more convenient to first present the construction of $w$ from $S$. 
Then we present the automaton $\Aa$ and analyse its runs on $w$.
 
Let $M=1+\max_{s\in S} |s|$.
By sorting $S$ we may assume that $S=\{s_1,s_2,\ldots,s_n\}$, where $0<s_1<\ldots<s_n<M$. 
We set $\Sigma=\{\step,\phase\}$. 
The word is defined as 
\<
  w = u\ \phase\ u\ \phase\ v,
\> 
where
\begin{eqnarray*}
  u & = &           2(M-s_n)\ \step\           2(s_n-s_{n-1})\ \step\           2(s_{n-1}-s_{n-2})\ \step\ \ldots\ \step\           2(s_2-s_1)\ \step\ 2(s_1-0);\\
  v & = & \phantom{2}(M-s_n)\ \step\ \phantom{2}(s_n-s_{n-1})\ \step\ \phantom{2}(s_{n-1}-s_{n-2})\ \step\ \ldots\ \step\ \phantom{2}(s_2-s_1)\ \step.
\end{eqnarray*}
Note that $w$ can be constructed from $S$ in time $\Oh(n\log n)$.
Intuitively, the factors $u$, $u$, and $v$ above are responsible for 
the choice of $a$, $b$, and $c$, respectively. 
We now describe a timed automaton $\Aa$ that accepts $w$ if and 
only if $a+b=c$. 

The automaton $\Aa$ is depicted in Figure~\ref{fig:3sum}. 
It uses two clocks, named $\mathtt{x}$ and $\mathtt{y}$.
Note that all the transitions have trivial (always true) clock conditions,
apart from the transition from $r_1$ to $r_2$, where we check that the sum 
of clock values is equal to $4M$.
The only initial state is $p_1$, the only accepting state is $r_2$.
 
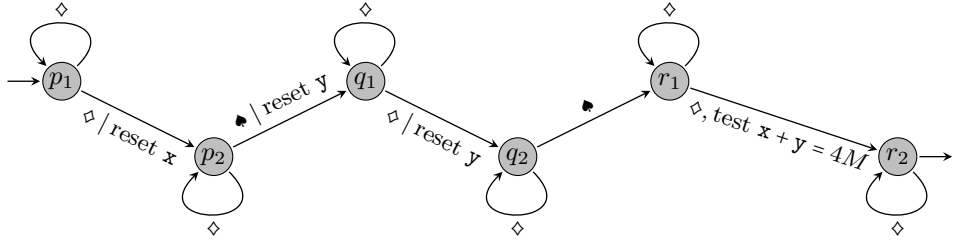
\begin{figure}
\centering
\begin{tikzpicture}
  \draw (0,0) node [biggray] (p1) {$p_1$};
  \draw (2,-1) node [biggray] (p2) {$p_2$};
  \draw (4,0) node [biggray] (q1) {$q_1$};
  \draw (6,-1) node [biggray] (q2) {$q_2$};
  \draw (8,0) node [biggray] (r1) {$r_1$};
  \draw (11,-1) node [biggray] (r2) {$r_2$};

  \draw [arrow] (-0.75,0) to (p1);
  \draw [arrow] (r2) to (11.75,-1);
  \draw [arrow] (p1) to [out=45, in=135, looseness=8] node [black, above=-0.05] 
        {\small $\step$} (p1);
  \draw [arrow] (p2) to [out=-45, in=-135, looseness=8] node [black, below=-0.05] 
        {\small $\step$} (p2);
  \draw [arrow] (q1) to [out=45, in=135, looseness=8] node [black, above=-0.05] 
        {\small $\step$} (q1);
  \draw [arrow] (q2) to [out=-45, in=-135, looseness=8] node [black, below=-0.05] 
        {\small $\step$} (q2);
  \draw [arrow] (r1) to [out=45, in=135, looseness=8] node [black, above=-0.05] 
        {\small $\step$} (r1);
  \draw [arrow] (r2) to [out=-45, in=-135, looseness=8] node [black, below=-0.05] 
        {\small $\step$} (r2);
  \draw [arrow] (p1) to node [black, sloped, below=-0.05] {\small $\step\mid\text{reset } \mathtt{x}$} (p2);
  \draw [arrow] (p2) to node [black, sloped, above=-0.05] {\small $\phase\mid\text{reset } \mathtt{y}$} (q1);
  \draw [arrow] (q1) to node [black, sloped, below=-0.05] {\small $\step\mid\text{reset } \mathtt{y}$} (q2);
  \draw [arrow] (q2) to node [black, sloped, above=-0.05] {\small $\phase$} (r1);
  \draw [arrow] (r1) to node [black, sloped, below=-0.05] {\small $\step,\text{test } \mathtt{x}+\mathtt{y}=4M$} (r2);
\end{tikzpicture}
\caption{Timed automaton for reducing {\sc 3SUM}.}\label{fig:3sum}
\end{figure}

We now analyse the runs of $\Aa$ on $w$, with the goal of showing that 
$\Aa$ accepts $w$ if and only if there are $a,b,c\in S$ such that $a+b=c$.
Consider any successful run $\rho$ of $\Aa$ on $w$.
Observe that the moment of reading the first symbol $\phase$ in $w$ 
must coincide with firing the transition from $p_2$ to $q_1$.
At this moment, the automaton has consumed the first factor $u$ of $w$,
and there was a moment where it moved from state $p_1$ to state $p_2$ 
upon reading one of the $\step$ symbols from $u$. 
Supposing that the transition in $\rho$ from $p_1$ to $p_2$ happens 
at the $i$-th symbol $\step$ of $u$, the clock valuation at the moment 
of reaching $q_1$ for the first time must satisfy
$\mathtt{x} = 2(s_i-s_{i-1})+\ldots+2(s_2-s_1)+2s_1$ ($=2s_i$)
and $\mathtt{y} = 0$.
We conclude the following.
 
\begin{claim}
The set of possible clock valuations at the moment of reaching 
the state $q_1$ for the first time is $\{(\mathtt{x}=2a, ~ \mathtt{y}=0) \::\: a\in S\}$.
\end{claim}

Next, observe that the moment of reading the second occurrence of $\spadesuit$ in $w$ 
must coincide with firing the transition from $q_2$ to $r_1$.
Between the first and the second symbol $\spadesuit$ the automaton consumes 
the second factor $u$, and during this the clock $\mathtt{x}$ increases exactly 
by the sum of the time spans within $u$, i.e.~by $2M$.
On consuming the second factor $u$, the clock $\mathtt{y}$ is reset once,
and precisely when firing the transition from $q_1$ to $q_2$, which happens
on reading one of the occurrences of $\step$ in $u$.
Again, if this happens when reading the $j$-th occurrence of $\step$, 
then, after the reset, $\mathtt{y}$ is incremented by exactly $2s_j$ units.
We conclude the following.
 
\begin{claim}
The set of possible clock valuations at the moment of reaching 
the state $r_1$ for the first time is $\{(\mathtt{x}=2a+2M, ~ \mathtt{y}=2b) \::\: a,b\in S\}$.
\end{claim}
 
Finally, after consuming the last factor $v$, 
the automaton can move to the accepting state $r_2$ if and only if at some point,
upon reading an occurrence of $\step$, the condition $x+y=4M$ holds.
Observe that the sum of the first $k$ numbers encoded in $v$ is equal 
to $M-s_{n-k+1}$. Hence, after parsing those numbers, the set of possible 
clock valuations is $\{(\mathtt{x}=2a+2M+M-c, ~ \mathtt{y}=2b+M-c) \::\: a,b\in S\}$,
for some choice of $c\in S$.
Moreover, the latter valuations satisfy the condition $\mathtt{x}+\mathtt{y}=4M$ 
if and only if $a+b=c$.

Based on the above arguments, we infer that a successful run like $\rho$ 
exists on input $w$ if and only if there are $a,b,c\in S$ such that $a+b=c$.
To conclude the proof, we observe that if an algorithm could decide whether 
$\Aa$ accepts $w$ in time $\Oh(n^{2-\delta})$ for any $\delta>0$,
then by combining this algorithm with the presented construction, 
one could solve {\sc{3SUM}} in time $\Oh(n^{2-\delta})$. This would contradict the {\sc{3SUM}} Conjecture.
\end{proof}

%Observe here that, as discussed in the description of the computation model in Section~\ref{sec:preliminaries}, 
%we may assume that the algorithm deciding whether $\Aa$ accepts $w$ works in the standard RAM model with words of bit length $\Oh(\log n)$ (as stipulated by the {\sc{3SUM}} Conjecture), 
%rather than the extended model that we use for handling floating-point numbers.
%This is because all clock constants in $\Aa$, as well as all numbers in $w$, are integers of bit length $\Oh(\log n)$.

Theorem \ref{th:two_clocks} now follows almost directly from the previous lemma. 
Consider the timed automaton $\Aa$ provided by Lemma~\ref{lem:beaver}. 
If a data structure as in the statement of the theorem existed, 
then using this data structure one could decide in strongly sub-quadratic time
whether any input timed word $w$ is accepted by $\Aa$, by simply applying the
sequence of $\read(\cdot)$ operations corresponding to $w$, followed by the query 
$\accepted()$.

\end{document}